\documentclass{lmcs}
\pdfoutput=1

\usepackage{lastpage}
\lmcsdoi{15}{3}{8}
\lmcsheading{}{\pageref{LastPage}}{}{}%
{Aug.~25,~2017}{Jul.~31,~2019}{}

\usepackage[utf8]{inputenc}
\usepackage{xargs}
\usepackage{ifthen}
\usepackage{etoolbox}
\usepackage{graphicx}
\usepackage{xspace}
\usepackage{amsmath, amsthm, amssymb}
\usepackage{multirow}
\usepackage{bold-extra}

\usepackage{silence}
\usepackage{subcaption}
\makeatletter
\let\c@table\c@figure
\let\ftype@table\ftype@figure
\makeatother

\usepackage{hyperref}
\usepackage{bookmark}

\usepackage{booktabs}

\usepackage[inline]{enumitem}
\newlist{itemizeinline}{itemize*}{2}
\setlist[itemizeinline]{before={ },afterlabel={},label=,itemjoin={;{ }{ }},itemjoin*={;{ }and{ }{ }}}

\newlist{subthm}{enumerate}{2}
\setlist[subthm]{itemsep={0.5ex plus 0.25ex minus 0.5ex},parsep=0pt,topsep={0.5ex plus 0.25ex minus 0.5ex},label={\upshape(\roman{*})},ref={\thethm(\roman{*})},left=-1.5ex}

\theoremstyle{plain}

\usepackage{ifthen}
\usepackage{etoolbox}
\usepackage{xargs}
\undef\clipbox%
\usepackage{adjustbox}
\usepackage{chemarrow}

\newcommandx{\newtheoremy}[3][2={}]{
  \ifthenelse{\equal{#2}{}}{
    \ifcsmacro{#1}{}{\newtheorem{#1}{#3}}
  }{
    \ifcsmacro{#1}{}{\newtheorem{#1}[#2]{#3}}
  }
}

\newcommand{\thmBlockFont}[1]{#1}

\newtheoremy{Theorem}{\thmBlockFont{Theorem}}
\newtheoremy{Corollary}[Theorem]{\thmBlockFont{Corollary}}

\makeatletter
\newcommand\ifcounter[3]{\@ifundefined{c@#1}{#3}{#2}}
\makeatother
\ifcounter{thm}{}{\newcounter{thm}}
\newtheoremy{algorithm}[thm]{\thmBlockFont{Algorithm}}
\newtheoremy{corollary}[thm]{\thmBlockFont{Corollary}}
\newtheoremy{conjecture}[thm]{\thmBlockFont{Conjecture}}
\newtheoremy{definition}[thm]{\thmBlockFont{Definition}}
\newtheoremy{example}[thm]{\thmBlockFont{Example}}
\newtheoremy{lemma}[thm]{\thmBlockFont{Lemma}}
\newtheoremy{proposition}[thm]{\thmBlockFont{Proposition}}
\newtheoremy{property}[thm]{\thmBlockFont{Property}}
\newtheoremy{question}[thm]{\thmBlockFont{Question}}
\newtheoremy{remark}[thm]{\thmBlockFont{Remark}}
\newtheoremy{notation}[thm]{\thmBlockFont{Notation}}
\newtheoremy{theorem}[thm]{\thmBlockFont{Theorem}}

\newtheoremy{claim}{\thmBlockFont{Claim}}[thm]
\renewcommand{\theclaim}{\thethm.\arabic{claim}}

\newtheoremy{falsepropositionX}{\thmBlockFont{Proposition}}

\newtheoremy{falsetheoremX}{\thmBlockFont{Theorem}}

\newtheoremy{falsecorollaryX}{\thmBlockFont{Corollary}}

\newtheoremy{falselemmaX}{\thmBlockFont{Lemma}}

\newtheorem*{falsestatementX}{\thmBlockFont{\thestatement}}
\newenvironment{falsestatement}[1]{\def\thestatement{#1}\begin{falsestatementX}}{\end{falsestatementX}\let\thestatement\relax}

\newcommand{\lcorollary}[1]{\label{c.#1}}

\newcommand{\ldefinition}[1]{\label{d.#1}}

\newcommand{\llemma}[1]{\label{l.#1}}
\newcommand{\lproblem}[1]{\label{pb.#1}}
\newcommand{\lproposition}[1]{\label{p.#1}}
\newcommand{\lproperty}[1]{\label{pp.#1}}

\newcommand{\lsection}[1]{\label{s.#1}}
\newcommand{\ltable}[1]{\label{t.#1}}
\newcommand{\lfigure}[1]{\label{f.#1}}
\newcommand{\ltheorem}[1]{\label{t.#1}}
\newcommand{\lequation}[1]{\label{eq.#1}}
\newcommand{\lclaim}[1]{\label{cl.#1}}
\newcommand{\preprocgenref}[2]{}
\newcommand{\generalref}[2]{%
  \preprocgenref{#1}{#2}%
  \ifthenelse{\equal{#1}{eq}}%
  {(\ref{#1.#2})}%
  {\ref{#1.#2}}%
}
\newcommand{\generalpageref}[2]{\pageref{#1.#2}}

\makeatletter
\newcommand*{\ralgorithm}{\@ifstar{\generalref{a}}{Algorithm~\ralgorithm*}}
\newcommand*{\palgorithm}{\@ifstar{\generalpageref{a}}{page~\palgorithm*}}

\newcommand*{\rcorollary}{\@ifstar{\generalref{c}}{Corollary~\rcorollary*}}
\newcommand*{\pcorollary}{\@ifstar{\generalpageref{c}}{page~\pcorollary*}}

\newcommand*{\rconjecture}{\@ifstar{\generalref{cj}}{Conjecture~\rconjecture*}}
\newcommand*{\pconjecture}{\@ifstar{\generalpageref{cj}}{page~\pconjecture*}}

\newcommand*{\rdefinition}{\@ifstar{\generalref{d}}{Definition~\rdefinition*}}
\newcommand*{\pdefinition}{\@ifstar{\generalpageref{d}}{page~\pdefinition*}}

\newcommand*{\rexample}{\@ifstar{\generalref{e}}{Example~\rexample*}}
\newcommand*{\pexample}{\@ifstar{\generalpageref{e}}{page~\pexample*}}

\newcommand*{\rlemma}{\@ifstar{\generalref{l}}{Lemma~\rlemma*}}
\newcommand*{\plemma}{\@ifstar{\generalpageref{l}}{page~\plemma*}}

\newcommand*{\rproblem}{\@ifstar{\generalref{pb}}{Problem~\rproblem*}}
\newcommand*{\pproblem}{\@ifstar{\generalpageref{pb}}{page~\pproblem*}}

\newcommand*{\rproposition}{\@ifstar{\generalref{p}}{Proposition~\rproposition*}}
\newcommand*{\pproposition}{\@ifstar{\generalpageref{p}}{page~\pproposition*}}

\newcommand*{\rproperty}{\@ifstar{\generalref{pp}}{Property~\rproperty*}}
\newcommand*{\pproperty}{\@ifstar{\generalpageref{pp}}{page~\pproperty*}}

\newcommand*{\rprocedure}{\@ifstar{\generalref{pc}}{Procedure~\rprocedure*}}
\newcommand*{\pprocedure}{\@ifstar{\generalpageref{pc}}{page~\pprocedure*}}

\newcommand*{\rremark}{\@ifstar{\generalref{r}}{Remark~\rremark*}}
\newcommand*{\premark}{\@ifstar{\generalpageref{r}}{page~\premark*}}

\newcommand*{\rnotation}{\@ifstar{\generalref{n}}{Notation~\rnotation*}}
\newcommand*{\pnotation}{\@ifstar{\generalpageref{n}}{page~\pnotation*}}

\newcommand*{\rsection}{\@ifstar{\generalref{s}}{Section~\rsection*}}
\newcommand*{\psection}{\@ifstar{\generalpageref{s}}{page~\psection*}}

\newcommand*{\rtable@}{\@ifstar{\generalref{t}}{Table~\rtable*}}
\newcommand*{\rtable}{\protect\rtable@}
\newcommand*{\ptable}{\@ifstar{\generalpageref{t}}{page~\ptable*}}

\newcommand*{\rfigure}{\@ifstar{\generalref{f}}{Figure~\rfigure*}}
\newcommand*{\pfigure}{\@ifstar{\generalpageref{f}}{page~\pfigure*}}

\newcommand*{\requation}{\@ifstar{\generalref{eq}}{Equation~\requation*}}
\newcommand*{\pequation}{\@ifstar{\generalpageref{eq}}{page~\pequation*}}

\newcommand*{\rtheorem}{\@ifstar{\generalref{t}}{Theorem~\rtheorem*}}
\newcommand*{\ptheorem}{\@ifstar{\generalpageref{t}}{page~\ptheorem*}}

\newcommand*{\rclaim}{\@ifstar{\generalref{cl}}{Claim~\rclaim*}}
\newcommand*{\pclaim}{\@ifstar{\generalpageref{cl}}{page~\pclaim*}}
\makeatother

\usepackage{boites}
\makeatletter
 \newdimen\bk@hauteurcourrante%
  \newdimen\bk@hauteursuivante%
  \newdimen\bk@tempdim%
\newenvironment{leftbar}{%
  \def\bk@espace{ }%
  \def\pt@to@bp##1{##1=.99627393548##1}
  \def\bkvz@before@breakbox{\ifhmode\par\fi\bk@hauteurcourrante=1200bp}%
  \def\bkvz@set@linewidth{\advance\linewidth-0.5\parindent}%
  \def\bkvz@left{\hskip 1pt\vrule\@width 0.5pt\hskip0.5\parindent\hskip -1.5pt}
  \let\bkvz@right\relax
  \let\bkvz@top\relax
  \let\bkvz@bottom\relax
  \breakbox}{\endbreakbox}

\newenvironment{proofwithbar}[1][Proof of the claim]{\begin{leftbar}\noindent{\itshape#1}.~}{\end{leftbar}}

\makeatother

\newcommandx{\wlen}[1]{|#1|}

\newcommandx{\cod}[2][2={}]{\ifthenelse{\equal{#2}{}}{\langle#1\rangle}{\langle#1\rangle_{#2}}}
\newcommandx{\floor}[1]{\lfloor#1\rfloor}
\newcommandx{\bfloor}[1]{\left\lfloor#1\right\rfloor}
\newcommandx{\bceil}[1]{\left\lceil#1\right\rceil}
\newcommandx{\ceil}[1]{\lceil#1\rceil}

\newcommandx{\newcommandy}[5][1=i,3=0,4={}]{%
  \ifthenelse{\isundefined{#2}}{\newcommandx{#2}[#3][#4]{#5}}{%
      \ifthenelse{\equal{#1}{i}}{}{}%
      \ifthenelse{\equal{#1}{o}}{\renewcommandx{#2}[#3][#4]{#5}}{}%
    }%
}

\newcommandx{\yrightarrow}[4][1=\empty, 2=\empty, 4=\empty, usedefault=@]{%
  \ifthenelse{\equal{#1}{\empty}}%
  {
    \xrightarrow{~\adjustbox{raise={-#4}{\height-#4}{0pt},trim=0pt 0pt 0pt 1pt}{\ensuremath{\scriptstyle#3}}~}%
  }{
    \adjustbox{trim=0pt 2pt 0pt 0pt}{\ensuremath{\xrightarrow
    [\,~\adjustbox{scale=0.9,raise={#2}{\height}}{\ensuremath{\scriptstyle#1}}~\,]
    {~\adjustbox{raise={-#4}{\height-#4}{0pt},trim=0pt 0pt 0pt 1pt}{\ensuremath{\scriptstyle#3}}~}}}%
  }%
}
\newcommandx{\ylefttarrow}[4][1=\empty, 2=\empty, 4=\empty, usedefault=@]{%
  \ifthenelse{\equal{#1}{\empty}}%
  {
    \xleftarrow{~\adjustbox{raise={-#4}{\height-#4}{0pt},trim=0pt 0pt 0pt 1pt}{\ensuremath{\scriptstyle#3}}~}%
  }{
    \adjustbox{trim=0pt #2 0pt 0pt}{\ensuremath{\xleftarrow
    [~\adjustbox{scale=0.9,raise={#2}{\height}}{\ensuremath{\scriptstyle#1}}~]
    {~\adjustbox{raise={-#4}{\height-#4}{0pt},trim=0pt 0pt 0pt 1pt}{\scriptstyle#3}~}}}%
  }%
}
\newcommandy[o]{\pathx}[2][2=\empty]{{\let\rightarrow\chemarrow%
  \nlb%
  \hspace*{1mm}\yrightarrow[#2][3pt]{\minwidthbox{$\scriptstyle#1$}{5mm}\hspace*{2pt}}[3.5pt]\hspace*{1mm}%
  \nlb%
}}
\newcommandy[o]{\pathy}[2][2=\empty]{{%
  \nlb%
  \hspace*{1mm}\ylefttarrow[#2][3pt]{\hspace*{2pt}\minwidthbox{$\scriptstyle#1$}{5mm}}[3.5pt]\hspace*{1mm}%
  \nlb%
}}
\newcommand*{\minwidthbox}[2]{%
  \makebox[{\ifdim#2<\width\width\else#2\fi}]{#1}%
}

\makeatletter

\newlength{\vm@xmd@d}
\newlength{\vm@xmd@n}
\newlength{\vm@xmd@s}
\newlength{\vm@xmd@ss}
\newcommandy[o]{\xmd}{%
  \ifmmode%
    \mathchoice%
    {\hspace*{\vm@xmd@d}}
    {\hspace*{\vm@xmd@n}}
    {\hspace*{\vm@xmd@s}}
    {\hspace*{\vm@xmd@ss}}
  \else%
    \hspace*{\vm@xmd@n}%
  \fi%
}
\makeatother

\newcommand{\val}[1]{\widebar{#1}}
\newcommand{\card}[1]{\texttt{Card}(#1)} 
\newcommand{\strong}[1]{\textbf{#1}}

\newcommand{\set}[1]{%
  \left\{\mathchoice%
  {\halfspace#1\halfspace}%
  {\thirdspace#1\thirdspace}%
  {#1}%
  {#1}\right\}%
}
\newcommand{\setq}[2]{\left\{\halfspace#1~\middle|~#2\halfspace\right\}}
\newcommandy[o]{\Z}{\mathbb{Z}}
\newcommandy[o]{\N}{\mathbb{N}}
\newcommandy[o]{\Q}{\mathbb{Q}}
\newcommandy[o]{\R}{\mathbb{R}}
\newcommand{\widebar}{\overline}

\newcommandy[o]{\Cup}{\bigcup}
\newcommand{\nlb}{\nolinebreak}

\renewcommand{\thmBlockFont}[1]{\ssc{#1}}

\makeatletter
\newcommandx{\newcommandWithStar}[3][1=i]{%
  \newcommandy[#1]{#2}{\protect\@ifstar{\leavevmode\protect\nlb$\protect#3$}{#3}}
}
\makeatother

\makeatletter
\newcommand{\vm@date@separator}{\hspace*{0.15ex}\rule[0.4\vm@date@height]{1ex}{0.07\vm@date@height}\hspace*{0.15ex}}
\newcommand{\vmdatefont}[1]{#1}
\newcommand{\isotoday}{%
  \vmdatefont{
    \newdimen\vm@date@height%
    \setbox0=\hbox{0123456789}%
    \vm@date@height=\ht0 \advance\vm@date@height by -\dp0 
    \the\year\vm@date@separator\two@digits{\month}\vm@date@separator\two@digits{\day}%
  }%
}

\makeatother

\newcommand{\vmfbox}[1]{{%
  \fboxsep=0pt%
  \ifmmode%
    \mathchoice%
      {\fbox{$\displaystyle#1$}}%
      {\fbox{$\textstyle#1$}}%
      {\fbox{$\scriptstyle#1$}}%
      {\fbox{$\scriptscriptstyle#1$}}%
  \else%
    \fbox{#1}%
  \fi%
}}

\newcommand{\halfspace}{\hspace{0.5\fontdimen2\font plus 0.5\fontdimen3\font 
minus 0.5\fontdimen4\font}}
\newcommand{\thirdspace}{\hspace{0.33\fontdimen2\font plus 0.33\fontdimen3\font 
minus 0.33\fontdimen4\font}}
\newcommand{\fourthspace}{\hspace{0.25\fontdimen2\font plus 0.25\fontdimen3\font 
minus 0.25\fontdimen4\font}}
\newcommand{\fifthspace}{\hspace{0.2\fontdimen2\font plus 0.2\fontdimen3\font 
minus 0.2\fontdimen4\font}}

\newcommand{\divides}{\halfspace|\halfspace}
\newcommand\Tstrut{\rule{0pt}{2.6ex}}         
\newcommand\Bstrut{\rule[-0.9ex]{0pt}{0pt}}   

\renewcommand{\thmBlockFont}[1]{#1}
\theoremstyle{plain}

\newtheorem{properties}[thm]{Properties}
\newtheorem*{runex}{{\normalfont\textit{Running example}}}
\renewcommand{\leq}{\leqslant}
\renewcommand{\geq}{\geqslant}
\renewcommand{\phi}{\varphi}
\renewcommand{\epsilon}{\varepsilon}
\renewcommand{\mod}{\text{~mod~}}

\newtheorem*{drawconv}{\thmBlockFont{Drawing Convention}}

\theoremstyle{plain}
\newtheorem*{theorem*}{Theorem}

\newcommand{\nplusun}{\hspace*{1.2pt plus 2pt}{+}\hspace*{0.6pt plus 2pt}1}
\newcommand{\nmoinsun}{\hspace*{1.2pt plus 2pt}{-}\hspace*{0.6pt plus 2pt}1}

\newcommand{\vmiminus}{\scalebox{0.75}[1.0]{\scriptsize$-$}}
\newcommand{\vmiplus}{\scalebox{0.75}[0.75]{\scriptsize$+$}}
\newcommand{\iplusun}{\hspace*{0.4pt}{\vmiplus}1}
\newcommand{\imoinsun}{\hspace*{0.4pt}{\vmiminus}1}

\newcommand{\po}{\mathchoice{+1}{\nplusun}{\iplusun}{+1}}
\newcommand{\mo}{\mathchoice{-1}{\nmoinsun}{\imoinsun}{-1}}

\newcommand{\trianglebullet}{\hspace*{.05ex}\raisebox{.1ex}{$\triangleright$}}

\newcommandx{\condfun}[2][2={\Ac}]{\ifstrempty{#1}{\gamma_{#2}}{\gamma_{#2}(#1)}} 
\newcommand{\cond}[1]{\boldsymbol{\mathit{CG}}(#1)} 

\newcommandWithStar{\Ab}{A_b}

\newcommandWithStar{\Abs}{{\Ab}^{\!\!*}}
\newcommandWithStar{\Gp}{\mathbb{G}_{\per}}

\newcommandx{\pascal}[2][1=R,2=\per]{\Pc_{#2}^{#1}}
\newcommandx{\pascalp}[2][1=R,2=\per]{{\Pc'}_{#2}^{#1}}

\newcommand{\ZZ}[1][\per]{\Z/#1\Z}
\newcommand{\ZZxZZ}[1][\per]{\ZZ\xmd{\times}\xmd\ZZ[\ord]}
\newcommandWithStar{\per}{p}
\newcommandWithStar{\base}{b}
\newcommandWithStar{\ord}{\psi}
\newcommand{\behav}[1]{L(#1)} 

\newcommand{\texteq}[1]{\hspace*{10mm minus 3mm}\text{#1}\hspace*{10mm minus 3mm}}

\renewcommand{\mod}{%
  \mathchoice%
    {\halfspace\scalebox{0.8}[0.7]{$\displaystyle{\normalfont\textsf{\%}}$}\halfspace}%
    {\halfspace\scalebox{0.8}[0.7]{$\textstyle{\normalfont\textsf{\%}}$}\halfspace}%
    {\fourthspace\scalebox{0.8}[0.7]{$\scriptstyle{\textsf{\%}}$}\fourthspace}%
    {\fifthspace\scalebox{0.8}[0.7]{$\scriptscriptstyle{\textsf{\%}}$}\fifthspace}%
  }

\newcommandWithStar{\EdR}{E_{d}^{R}}
\newcommandWithStar{\EpR}{R+p\xmd\N}
\newcommandWithStar{\EpRm}{E_{\per,\,\geq m}^{R}} 
\newcommandWithStar{\EpRmI}{I\cup E_{\per,\,\geq m}^{R}} 
\newcommandWithStar{\EpRI}{I\oplus(\EpR)}

\newcommandWithStar{\AdR}{\Ac_{d}^{R}}
\newcommandWithStar{\ApR}{\Ac_{\per}^{R}}
\newcommandWithStar{\ApRm}{\Bc_{\per,\,\geq m}^{R}} 

\newcommandWithStar{\FpR}{F_{\per}^{R}}
\newcommandWithStar{\Gm}{\mathcal{G}_{\geq m}} 
\newcommandWithStar{\GI}{\mathcal{G}_{I}}

\newcommand{\Nd}{\N^{d}}
\newcommand{\SL}{\mathsf{S}_L}

\newlength{\pplen}
\newcommand{\vmppspace}{\setlength{\pplen}{0.25em plus 0.1em minus 0.05em}\hspace{\pplen}}
\newcommand{\brec}{$b$-recognisable\xspace}
\newcommand{\UP}{\texorpdfstring{{$\mathbb{UP}$}\xspace}{UP}}
\newcommand{\up}{u.p.\@\xspace}
\newcommand{\upf}{u.p. }
\newcommand{\aupsn}{a \upsn}
\newcommand{\upsn}{u.p{.}\vmppspace subset of~$\N$\xspace} 
\newcommand{\upsns}{u.p{.}\vmppspace subsets of~$\N$\xspace} 
\newcommand{\upaut}{automaton in \UP}

\newcommand{\scc}[1]{s.c.c.\ifthenelse{\equal{#1}{.}}{}{\@\xspace#1}}
\newcommand{\sccs}{s.c.c.'s\xspace}
\newcommand{\ascc}{an \scc}
\renewcommand{\dag}{d.a.g.\@\xspace}

\newcommandx{\repr}[2][2={\base}]{\langle\xmd#1\xmd\rangle}

\renewcommand{\val}[1]{\widebar{\thirdspace\rule{0pt}{1.4ex}#1\thirdspace}}

\newcommand{\tauu}[1][u]{\tau_{#1}}
\newcommand{\lefthk}[1][(h,k)]{\gamma_{#1}}
\newcommand{\lefthksub}[1][(h,k)]{\sigma_{#1}}
\newcommand{\Ahk}[1][(h,k)]{\Ac_{#1}}
\newcommand{\Qhk}[1][(h,k)]{Q_{#1}}
\newcommand{\dhk}[1][(h,k)]{\delta_{#1}}
\newcommand{\Fhk}[1][(h,k)]{F_{#1}}
\newcommand{\gx}{\mathbin{\diamond}}
\newcommand{\ZZxrZZ}{\ZZ\rtimes\ZZ[\ord]}

\newcommand{\bpp}[1]{\big((#1)\big)}
\newcommand{\bp}[1]{\big(#1\big)}

\newcommand{\aut}[1]{\left\langle#1 \right\rangle}
\newcommand{\Ac}{\mathcal{A}}
\newcommand{\Bc}{\mathcal{B}}
\newcommand{\Cc}{\mathcal{C}}
\renewcommand{\Mc}{\mathcal{M}}
\newcommand{\Uc}{\mathcal{U}}
\newcommand{\Pc}{\mathcal{P}}
\newcommand{\ap}{\mathbin{\boldsymbol\cdot}}

\newcommand{\descendant}{descendant\xspace}

\newcommand{\AcS}[1][S]{\Uc_{#1}}
\newcommand{\dX}{\Delta}

\newcommand{\cf}{\textit{cf}.\@\xspace}
\newcommand{\ie}{\textit{i.e}.\@\xspace}
\newcommand{\Cf}{\mathbf{C}}
\newcommand{\Df}{\mathbf{D}}

\newcommand{\Qf}{\mathbf{Q}}
\newcommand{\Ff}{\mathbf{F}}
\newcommand{\Kf}{\mathbf{K}}

\newcommand{\bigo}[1]{\mathsf{O}\hspace*{-1pt}\left(#1\right)}
\newcommand{\powerset}{\boldsymbol{\mathcal{P}}}

\newcommand{\UPatom}{\UP-atomic\xspace}

\newenvironment{claimproof}[1][Proof of Claim~\theclaim{}]%
{\begin{proofwithbar}[#1]}%
{\end{proofwithbar}}

\newcommand{\otop}[1][p]{\{0,\ldots,#1\mo\}}

\newcommand{\quantvrg}{,~~}
\newcommand{\quantsp}{\quad}
\newcommand{\eqpnt}{~.}
\newcommand{\eqpntvrg}{~;}
\newcommand{\eqvrg}{~,}

\newcommand{\AutScale}{0.55}

\begin{document}

\title[Deciding periodicity of \texorpdfstring{\MakeLowercase{$b$}}{b}-recognisable sets using LSDF convention]{An efficient algorithm to decide periodicity of \texorpdfstring{\MakeLowercase{$b$}}{b}-recognisable sets using LSDF convention{\rsuper*}}
\titlecomment{{\lsuper*} An early version of this work was published in the proceedings of the DLT conference~\cite{MarsSaka13a}; most of the results are also part of the  of the Ph.D thesis of the author~\cite{Mars16}.}
\author[V.~Marsault]{Victor Marsault}
\address{%
  LIGM,
  Universit{\'e} Paris-Est Marne-la-Vall{\'e}e,
  ESIEE Paris,
  {\'E}cole des Ponts ParisTech,
  CNRS,
  France
}
\address{%
  Laboratory for Foundations of Computer Science,
  School of Informatics,
  University of Edinburgh,
  United Kingdom
}
\address{%
  Department of Mathematics,
  Universit{\'e} de Li{\`e}ge,
  Belgium{\rsuper\dagger}
 }
 \thanks{%
  {\lsuper\dagger}While the author was affiliated with the University of Li{\`e}ge, he was
  supported by a Marie Sk{\l}odowska-Curie fellowship, partially funded by the
  European Union.
 }
\address{%
  IRIF,
  Universit{\'e} Paris-Diderot,
  France\vspace{1em}
}
\email{victor.marsault@u-pem.fr}

\graphicspath{{vcsg/}{.}}


%
\begin{abstract}\normalsize
  Let~$b$ be an integer strictly greater than~$1$.
  Each set of nonnegative integers is represented in base~$b$ by a
  language over~$\set{0,1, \dots ,\linebreak[0]b\mo}$.
  The set is said to be $b$-recognisable if it is represented by a
  regular language.
  It is known that ultimately periodic sets are $b$-recognisable, for
  every base $b$, and Cobham's theorem implies the converse: no other
  set is $b$-recognisable in every base~$b$.
  We consider the following decision problem: let~$S$ be a set of
  nonnegative integers that is $b$-recognisable,  given as a finite
  automaton over $\set{0,1, \dots,\linebreak[0]b\mo}$, is~$S$ periodic?
  Honkala showed in 1986 that this problem is decidable.
  Later on, Leroux used in 2005 the convention to write number
  representations with the least significant digit first (LSDF), and
  designed a quadratic algorithm to solve a more general problem.
  We use here LSDF convention as well and give a structural description
  of the minimal automata that accept periodic sets.
  Then, we show that it can be verified in linear time if a
  minimal automaton meets this description.
  In general, this yields a $O(b\xmd n\xmd $log$(n))$ procedure to
  decide whether an automaton with $n$ states accepts an
  ultimately periodic set of nonnegative integers.
\end{abstract}

\maketitle

\clearpage


\section{Introduction}

Let~$b$ be a fixed integer strictly greater than~$1$, called the \emph{base}.
Every nonnegative integer~$n$ is \emph{represented} (in base~$b$) by a
\emph{word}~$u$ over the digit
alphabet~${\Ab=\set{0,1,\ldots,b\mo}}$, and representation is unique up
to leading~$0$'s.
Hence, \emph{subsets} of~$\N$ are represented by \emph{languages}
of~$\Abs$.
Depending on the base, a given subset of~$\N$ may be represented by
a simple or complex language:
the set of powers of~$2$ is represented in base~$2$ by the regular
language~$10^*$; whereas it is
represented in base~$3$ by a language that is not context-free.
A subset of~$\N$ is said to be \emph{\brec} if it is
represented by a regular (or rational, or recognisable) language
over~$\Ab$.
On the other hand, a subset of~$\N$ is said \emph{recognisable} if it is,
via the identification of~$\N$ with~$a^{*}$
($n \leftrightarrow a^{n}$),
a regular language of~$a^{*}$.
A subset of~$\N$ is recognisable if and only if it is
\emph{ultimately periodic} (\up) and we use the latter terminology in
the sequel as it is both meaningful and more distinguishable from \brec.
It is common knowledge that every \up set (of nonnegative integers) is
\brec for every~$b$.
However, a~{$b$-recognisable} set for
some~$b$ is not necessarily \up, nor~$c$-recognisable for some other~$c$;
the set of all powers of~$2$, previously discussed, is an example of
these two facts.
It is a simple exercise to show that if~$b$ and~$c$ are \emph{multiplicatively
dependent} (that is, if there exist positive integers~$k$ and~$\ell$ such that
$b^{k}=c^{\ell}$), then every \brec set is a~$c$-recognisable set as
well.
A converse of these two properties is the theorem of Cobham~\cite{Cobh69}:
\emph{a set of numbers that is both~$b$- and~$c$-recognisable, for
multiplicatively independent~$b$ and~$c$, is u.p}.
It is a strong and deep result whose
proof is difficult (see also~\cite{BruyEtAl94,DuraFRigo11}).
After Cobham's theorem, another natural question on~\brec sets is the
decidability of periodicity.
It was positively solved in~1986:

\begin{theorem*}[Honkala~\cite{Honk86}]\ltheorem{hon}
  It is decidable whether an automaton over~$\Ab$ accepts an ultimately periodic
  set.
\end{theorem*}

The complexity of the decision procedure is not
an issue in the original work.
Neither are the properties or the structure of automata accepting
\up sets.
Given an automaton~$\Ac$, Honkala shows that there are bounds on the
parameters of the potential \up set accepted by~$\Ac$.
The property is then decidable as it is possible to enumerate all
automata that accept sets with smaller parameters and check whether
any of them is equivalent to~$\Ac$.
As detailed below, subsequent works on automata and number
representations brought some answers regarding the complexity
of the decision procedure, explicitly or implicitly.
In the present article, we follow the convention that number
representations are written least significant digit first (LSDF
convention) and show the following.
\newcommand{\sttcomplx}{%
  Let~$b>1$ be an integer.
  We assume that number representations are written in base~$b$ and
  with the least significant digit first.
  Given a minimal DFA~$\Ac$ with~$n$ states,
  it is decidable in time~$\bigo{b\xmd n}$ whether~$\Ac$ accepts an
  ultimately periodic set.
}
\begin{thm}\ltheorem{com-plx}%
  \sttcomplx%
\end{thm}
%


%
\newcommand{\sttcomplxnonmin}{%
  Given a DFA~$\Ac$ with~$n$ states, it is decidable in time~$\bigo{(b\xmd n)\xmd \log n}$
  whether~$\Ac$ accepts an ultimately periodic set.
}
\begin{cor}\lcorollary{com-plxnonmin}
  \sttcomplxnonmin%
\end{cor}
%


\subsection*{On the order of digits}

Honkala's problem gives birth to two different
problems when one writes either the least or the most significant digit
first (LSDF or MSDF, respectively).
These two problems are not polynomially equivalent.
In order to transform an instance~$\Ac$ of one of the problem into an instance
of the other, one must run on~$\Ac$ a transposition and then a determinisation.
This potentially leads to an exponential blow-up of the number of states.
This event occurs for the problem at hand for example with
the language~$L_n=\,1\,{(0+1)}^n\, 1\, {(0+1+\varepsilon)}^n0^*$
and its mirror~$K_n$.
The number of states in the minimal automaton accepting $L_n$ (resp.~$K_n$)
grows linearly (resp.\@ exponentially) with~$n$.
Evaluating~$L_n$ as LSDF encodings or~$K_n$ as MSDF encodings yields
the same finite (thus u.p.) set.
A recent work by Boigelot et al.~\cite{BoigEtAl17}
gives a quasi-linear algorithm to solve Honkala's problem when number
representations are written MSDF\@.
As noted above, this result cannot be used to solve efficiently
the problem using LSDF convention, which is the object
of the present paper.

\subsection*{Related work in the multidimensional setting}

New insights on Honkala's problem were obtained when stating it in
a higher dimensional space.
Let~$\N^{d}$ be the additive monoid of~$d$-tuples of nonnegative
integers.
Every~$d$-tuple in~$\N^{d}$ can be represented in base~$b$ by a
$d$-tuple of words over~$\Ab$ of \emph{the same length}, as shorter
words can be padded by~$0$'s without changing the corresponding value.
Such~$d$-tuples can be read by (finite) automata
over~${({\Ab}^{\! d})}^{*}$ --- automata reading on~$d$ synchronised
tapes ---
and a subset of~$\N^{d}$ is \brec if the set of the
$b$-representations of its elements is accepted by such an automaton.
On the other hand, the recognisable and rational subsets of~$\Nd$ are
defined in the classical way.
A subset of~$\Nd$ is \emph{recognisable} if it is saturated by a
congruence of finite index,
and is \emph{rational} if it may be expressed by a rational
expression.
If~$d=1$, then~$\Nd=\N$ is a free monoid and the family of rational sets is equal
to the family of recognisable sets; in this case, they are typically called
\emph{regular languages} via the identification of~$\N$ with~$a^{*}$.
Otherwise,~$\Nd$ is not a free monoid and the two families do not
coincide (\cf~\cite{Saka09}).
It is also common knowledge that every rational set of~$\Nd$ is \brec for
every~$b$, and the example in dimension~$1$ is enough to show that a \brec set
is not necessarily rational.
Semenov showed a generalisation of Cobham's theorem (\cf~\cite{Seme77,BruyEtAl94,DuraFRigo11}):
\emph{a subset of~$\Nd$ which is both~$b$- and~$c$-recognisable, for
multiplicatively independent~$b$ and~$c$, is rational}.
The generalisation of Honkala's theorem went as smoothly.

\begin{theorem*}[Muchnik~\cite{Much03}] \ltheorem{muc}%
  It is decidable whether a \brec subset of~$\Nd$ is rational.
\end{theorem*}

\begin{theorem*}[Leroux~\cite{Lero05}]
\ltheorem{ler}%
  Assuming that number representations are written LSDF, it is decidable in
  polynomial time whether a \brec subset of~$\Nd$ is rational.
\end{theorem*}

Muchnik's algorithm is triply exponential while
Leroux's is quadratic.
This improvement is based on sophisticated geometric constructions that
are detailed in~\cite{Lero06}.
Note that Leroux's result, restricted to dimension~$d=1$, readily yields a
quadratic procedure for Honkala's original problem.
The improvement to quasilinear complexity that we present here
(\rcorollary{com-plxnonmin}) is not due to a natural simplification of Leroux's
construction for the case of dimension~1.
Rational sets of~$\Nd$ have been characterised by Ginsburg and
Spanier~\cite{GinsSpan66} as sets definable in \emph{Presburger arithmetic}
(that is, definable by a formula of the first order logic with addition, denoted by~$FO[\N,+]$).
%
%
On the other hand, the Büchi-Bruyère theorem (\cf~\cite{Buchi1960,Bruy85,BruyEtAl94}) characterises \brec subsets
of~$\Nd$: \emph{A subset of~$\Nd$ is \brec if and only if it is definable by a formula of~$FO[\N,+,V_b]$}.
(The function $V_b:\N\rightarrow\N$ maps each~$n$ to the greatest power of~$b$ that divides~$n$.)

Using these two results, one may see that Muchnik's problem can (and was indeed)
stated in terms of logic: decide whether a formula of~$FO[\N,+,V_b]$ has an
equivalent formula in~$FO[\N,+]$.
However, the two statements are not equivalent
for complexity issues.
Using the Büchi-Bruyère Theorem to build an automaton from a formula may
give rise to a multi-exponential blow-up of the size.

\subsection*{Related work in non-standard numeration systems}

Generalisation of base~$p$ by nonstandard numeration systems gives an
extension of Honkala's problem, best expressed in terms of abstract numeration
systems.
Given a totally ordered alphabet~$A$, any language~$L\subseteq A^{*}$ defines
an \emph{abstract numeration system} (a.n.s.)~$\SL$ in which the
integer~$n\in\N$ is represented by the~$(n\po)$-th word of~$L$ in the radix
order (\cf~\cite{LecoRigo10-ib}).
The a.n.s.\@ is said to be regular if~$L$ is.
A subset of~$\N$ is called~$\SL$-recognisable if its representation
in the a.n.s.\vmppspace$\SL$ is a regular language.
It is known that
every \upsn is~$\SL$-recognisable for every regular a.n.s.\vmppspace$\SL$.
The extended Honkala's problem takes as input an~$\SL$-recognisable
set~$X$ and consists in deciding whether~$X$ is \upf%
It was observed in~\cite{AlloEtAl09,CharEtAl12} that, for a subset of~$\N$
the \emph{property of being \up}{ }is definable by a formula of the
Presburger arithmetic.
Hence, if~$\SL$ is a regular a.n.s.\@ in which addition is realised by
a finite automaton, then the extended Honkala's problem is decidable.
In particular, this approach solves the case where the numeration system is
a Pisot U-system (\cf~\cite{FrouSaka10-ib}).
On the other hand, with a proof similar to the one from the original
Honkala's paper, the problem was also shown to be decidable for a large
class of U-systems~\cite{BellEtAl09,Charlier2009}.
This class is incomparable with the
class of Pisot U-systems.
Finally, it is shown in~\cite{RigoMaes02,LecoRigo10-ib} that the
extended Honkala's problem is equivalent to deciding whether an HD0L
sequence is periodic (\cf~\cite{AlloShal03}).
Since then, this latter problem has been shown to be
decidable~\cite{Dura13,Mitr11}.
Hence, the extended Honkala's problem is also decidable in general.

\bigskip

These extensions were mentioned for the sake of completeness.
The present article is focused on solving the original problem
of Honkala when using LSDF convention.

\subsection*{Outline}

As it is often the case, the linear complexity of our algorithm
is obtained as the consequence of a structural characterisation.
After preliminaries, \rsection{minim} defines and study the class \UP of
the minimal automata that accept \up sets.
Then, we describe in \rsection{UP} a set of structural properties about the shapes and
positions of the strongly connected components (\sccs) and show that these
properties characterise the class \UP (\rtheorem{UP-stru-char}).
Finally, \rsection{line-comp} gives the linear algorithm underlying \rtheorem{com-plx},
which decides whether a given minimal automaton accepts a \up set.
The delicate part is to obtain a linear complexity in the special case where
the input automaton is strongly connected.
%




\section{Preliminaries}
\lsection{prelim}

\subsection{On automata}
\lsection{auto-defi}
An \emph{alphabet}~$A$ is a finite set of symbols, or \emph{letters};
in our case, letters will always be digits and the term \emph{digit}
will be used as a synonym of \emph{letter}.
We call \emph{word} over~$A$ a finite sequence of letters taken in~$A$; the
empty word is denoted by~$\epsilon$ and the length of a word~$u=a_0a_1\cdots
a_{k\mo}$ by~$\wlen{u}=\wlen{a_0a_1\cdots a_{k\mo}}=k$.
The set of words over~$A$ is denoted by~$A^*$, and a subset of~$A^*$
is called a \emph{language over~$A$}.
In this article, we consider only automata that are deterministic and finite.
Thus, an automaton is denoted by~$\Ac=\aut{A,Q,i,\delta,F}$,
where~$A$ is the \emph{alphabet},%
~$Q$ is the finite set of \emph{states},%
~$i\in Q$ is the \emph{initial state},%
~$F\subseteq Q$ is the set of \emph{final states},
and~${\delta:Q\times A \rightarrow Q}$ is the \emph{transition function}.
As usual,~$\delta$ is extended to a function~$Q\times A^*\rightarrow Q$
by~$\delta(q,\epsilon)=q$
and~$\delta(q,ua)\nlb=\nlb\delta(\delta(q,u),a)$.
When the context is clear,~$\delta(s,u)$ will also be denoted by~$s\ap u$.
A \emph{transition in~$\Ac$} is an element~$(s,a,t)$ in~$Q\times\Ab\times Q$ such
that~$\delta(s,a)=t$; it is usually denoted by~$s\pathx{a}[\Ac] t$ or
simply~$s\pathx{a} t$ when~$\Ac$ is clear from context.
A \emph{path in~$\Ac$} is a sequence of
transitions~$s_0\pathx{a_0}[\Ac]s_1\cdots \pathx{a_k}[\Ac]s_{k\po}$ which is
also denoted by~$s_0\pathx{u}[\Ac] s_{k\po}$ where~$u=a_0\xmd \cdots a_k$, and
we call~$s_0$ the \emph{origin},~$u$ the \emph{label} and~$s_{k+1}$ the
\emph{destination} of this path.
Note that this path exists if~$\delta(s_0,u)=s_{k+1}$.
We call \emph{run} any path originating from the initial state, and
\emph{the run of a word~$u$} refers to the run  labelled by~$u$ if it
exists; this path is well defined since our automata are
deterministic.
A word~$u$ in~$A^*$ is \emph{accepted} by~$\Ac$ if its run ends in a final
state, that is, if~$(i \cdot u)$ exists and belongs to~$F$.
The language accepted by~$\Ac$ is denoted by~$\behav{\Ac}$.
If every word has a run,~$\mathcal A$ is said to be \emph{complete}.
A state~$r$ is said \emph{reachable from} another state~$s$ if there
exists a path
from~$r$ to~$s$, and simply \emph{reachable} if it is reachable
from the initial state.
An automaton is said \emph{reachable} if all its states are \emph{reachable}.
%
%

\begin{drawconv}
  In figures, most automata will be over two-letter alphabets ($\set{0,1}$ or~$\set{0,g}$).
  For the sake of clarity, we omit labels in such cases:
  transitions labelled by~1 will be drawn with a thick line,
  those labelled by~0 with a thin line,
  and those by~$g$ with a double line.
\end{drawconv}

\begin{defi}
  Let~$\Ac$ and~$\Mc$ be two automata over the same alphabet~$A$
  \begin{subthm}
    \item \ldefinition{auto-morp}
    An \emph{(automaton) morphism} is a surjective
    function~$\phi:Q_{\Ac}\rightarrow Q_{\Mc}$
    that meets the following three conditions.
    \begin{subequations}
    \begin{gather}\allowdisplaybreaks
      \phi(i_{\Ac}) = \phi(i_{\Mc}) \lequation{auto-morp-init}\\
      \phi^{\mo}(F_{\Mc}) = F_{\Ac} \lequation{auto-morp-final}\\
      \lequation{auto-morp-trans}
      \forall a\in A \quantvrg \forall s\in Q_{\Ac}
        \quantsp \phi(s\cdot a) = \phi(s)\cdot a
    \end{gather}
    \end{subequations}

    \item
    If~$\phi$ denotes a morphism, we say that two states~$s$
    and~$s'$ are~\emph{$\phi$-equivalent} if they have the same
    image by~$\phi$.

    \item
    If there exists a morphism~$\Ac\rightarrow\Mc$, we say
    that~$\Mc$ is a \emph{quotient} of~$\Ac$.

  \end{subthm}
\end{defi}
\noindent
Given a regular language~$L$, it is classical (\cf~\cite{Saka09}, for instance)
that there exists a \emph{minimal automaton}~$\Mc$ that accepts~$L$:
it is the complete automaton accepting~$L$ with the minimal amount of states.
Moreover, given an automaton~$\Ac$ that accept~$L$, $\Mc$ may be computed in
quasi-linear time from~$\Ac$ and~$\Mc$ is a quotient of~$\Ac$.
\begin{defi}
  The transition monoid~$T$ of an automaton~$\Ac$ is the set of the
  functions induced by all words in~$A^*$  on the states of~$\Ac$:
  \begin{equation*}
    T = \setq{ \begin{array}{lcll} f_u:&Q_\Ac & \longrightarrow & Q_\Ac \\
                                  &q     & \longmapsto & q\cdot u
                \end{array}}{ u\in A^* } \eqpnt
  \end{equation*}
\end{defi}
Note that in the previous definition, since~$Q_\Ac$ is finite, there is
a finite number of
functions~$Q_\Ac \rightarrow Q_\Ac$ hence~$T$ is always finite.
\begin{defi} \ldefinition{grou-auto}
  An automaton~$\Ac$ over an alphabet~$A$ is called a \emph{group
  automaton} if every state of~$\Ac$ has a unique incoming and a unique
  outgoing transition labelled by each letter of~$A$.
\end{defi}

It follows from \rdefinition{grou-auto} that an automaton is a group automaton
if and only if its transition monoid is a group.
Moreover, that property is stable by quotient:

\begin{pty} \lproperty{quot-grou-auto}
  Every quotient of a group automaton is a group automaton.
\end{pty}

\subsection{On strongly connected components}

Two states~$s,s'$ of an automaton~$\Ac$ are \emph{strongly connected} if~$\Ac$
contains a path from~$s$ to~$s'$ and a path from~$s'$ to~$s$.
This defines an equivalence relation whose classes are called the \emph{strongly
connected components} (\sccs) of~$\Ac$.
Every state~$s$ of~$\Ac$ then belongs to a unique \scc.
Note that \ascc does not necessarily contains a circuit.
Indeed the \scc of an \emph{isolated state}~$s$ (that is, a state that do not belong to any
circuit), is the singleton~$\set{s}$ and is said \emph{trivial}.
Figures \rfigure*{cond-befo} and \rfigure*{cond-sccs} show an
automaton and its \sccs.
\begin{figure}[ht]
  \centering
  \begin{subfigure}[b]{0.4\textwidth}
    \centering
    \includegraphics[scale=\AutScale]{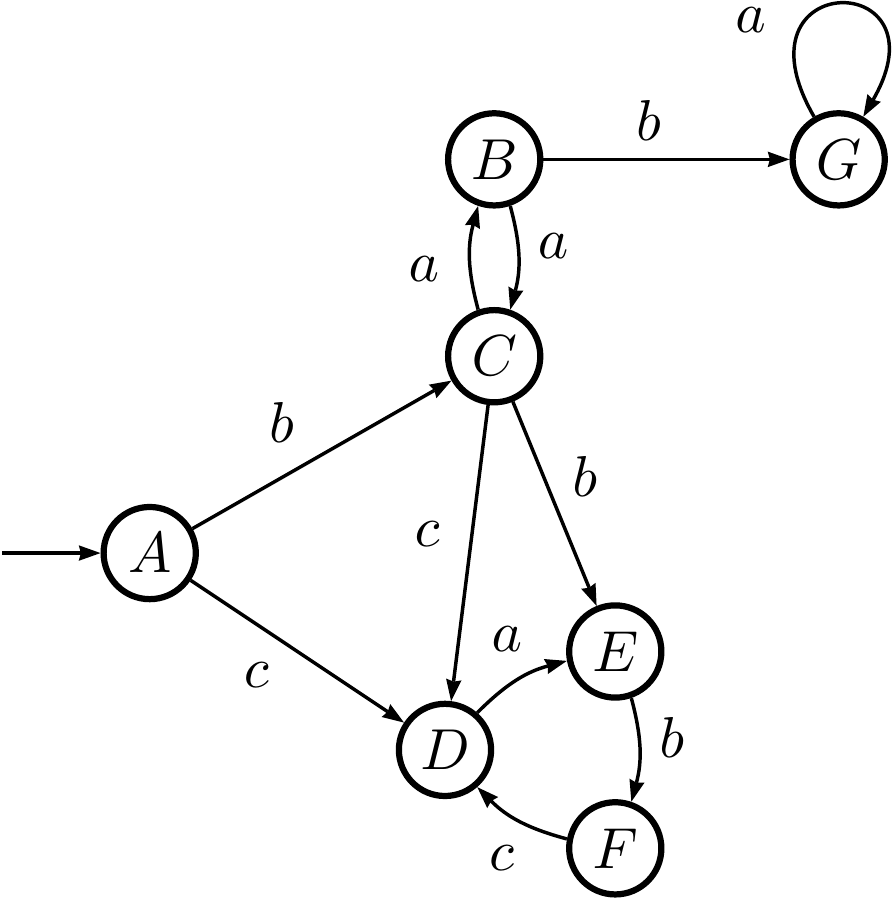}
    \caption{An automaton~$\Ac_1$}\lfigure{cond-befo}%
  \end{subfigure}
  \begin{subfigure}[b]{0.59\textwidth}
    \centering
    \includegraphics[scale=\AutScale]{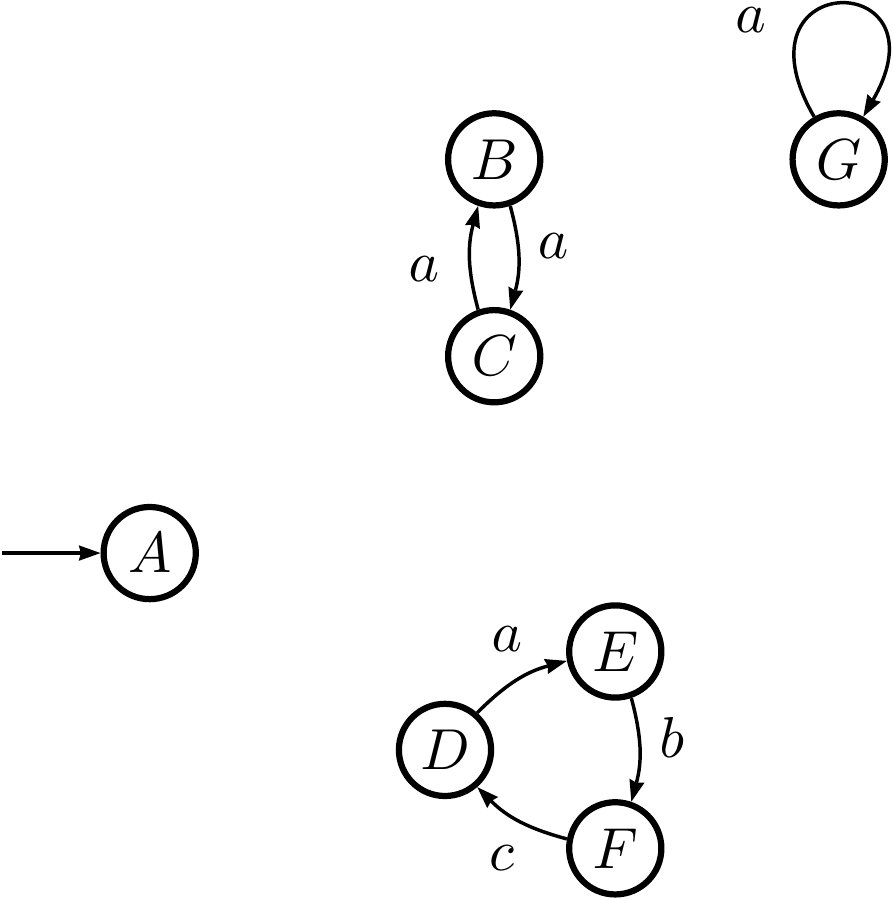}%
    \caption{The \sccs of~$\Ac_1$, and their internal transitions}\lfigure{cond-sccs}%
  \end{subfigure}

  \vspace{1.5em}
  \begin{subfigure}[b]{\textwidth}
    \centering
    \includegraphics[scale=\AutScale]{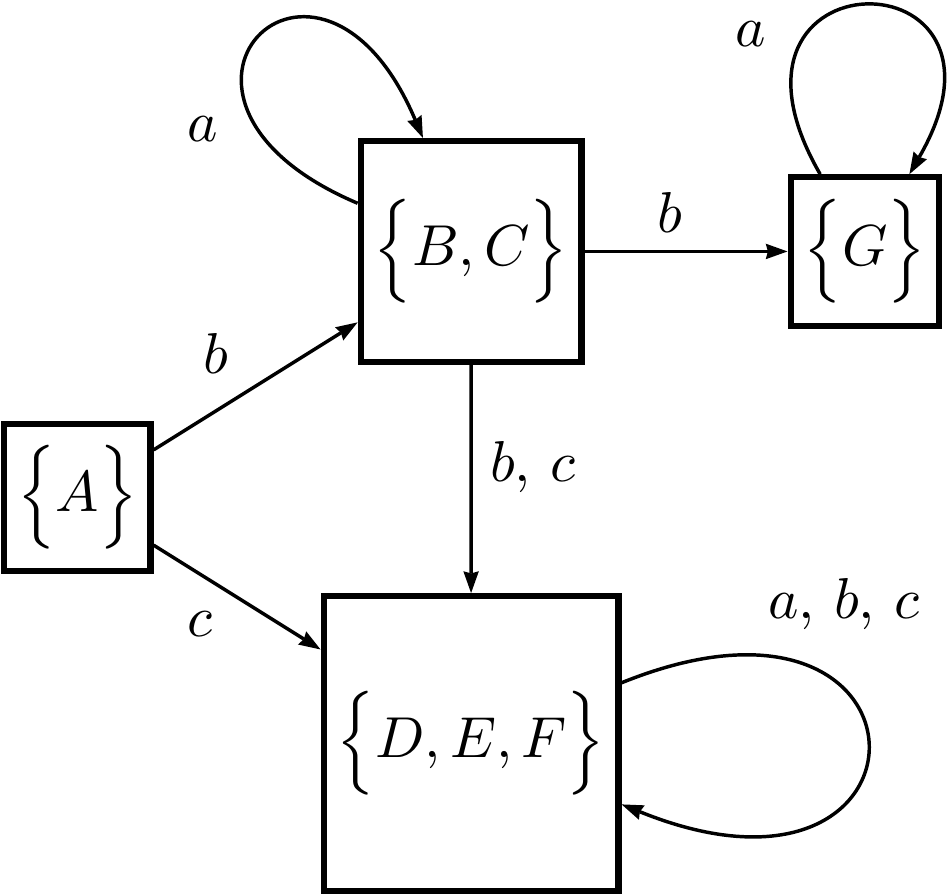}
    \caption{$\cond{\Ac_1}$, the component graph of~$\Ac_1$}\lfigure{cond-afte}%
  \end{subfigure}
  \caption{An automaton, its \sccs and its component graph.}
\end{figure}
The \emph{component graph}~$\cond{\Ac}$ of a an automaton~$\Ac$ is the labelled \dag
(directed  acyclic graph) that results from contracting each \scc into
a single vertex.
For instance, \rfigure{cond-afte} shows the component graph of~$\Ac_2$.
We say that \ascc~$X$ is a \emph{\descendant} of another \scc~$Y$ if
$X$ is a successor of~$Y$ in the component graph that is, if there is~$x\in X$ and~$y\in Y$
such that~$x\pathx{a} y$, for some letter~$a$.
It is classical that the component graph can be computed efficiently (\cf~\cite{CormEtAl09}), as stated below.
\begin{thm}
  \ltheorem{tarj}
  The component graph of an $m$-transitions automaton can be computed in
  time~$O(m)$.
\end{thm}
%

\subsection{On integer base numeration system}
\lsection{inte-base-defi}

Let~$\base$ be an integer strictly greater than~$1$ called the
\emph{base}.
It will be fixed throughout the article.
We briefly recall below the definition and elementary properties of
base-$b$ numeration systems.
Note that we represent numbers with the Least Significant Digit First (LSDF) a
convention used with some success in the past, for instance by Leroux~\cite{Lero05,Lero06}.

\medskip

Given two \emph{positive integers}~$n$ and~$m$, we denote by~$n\div m$ and~$n
\mod m$ respectively the quotient and the remainder of the Euclidean division
of~$n$ by~$m$, \ie~${n=(n\div m)\xmd m + (n\mod{m})}$ and~$0\leq
(n\mod{m}) < m$.
We index the letters of a word~$u$ from left to
right:~$u=a_0\xmd a_1 \xmd \cdots \xmd a_n$.

%

\medskip

Given a word~$u=a_0\xmd a_1 \xmd \cdots \xmd a_n$ over the
alphabet~$\Ab=\set{0,1,\ldots, b\mo}$,
its \emph{value} (in base~$\base$), denoted by~$\val{u}$, is given by the following expression.
\begin{equation}
  \lequation{inte-eval}
  \val{u} = \val{a_0\xmd a_1 \xmd \cdots \xmd a_n}
          = \sum_{i=0}^{n} a_i \xmd \base^i
\end{equation}
Words whose values are equal to some integer~$k$ are called
\emph{\base*-expansions} of~$k$.
Exactly one among them does not
\strong{end} with the digit 0; it is called the
\emph{\base*-representation} of~$k$, and is denoted by~$\repr{k}$.
We recall below formulas for evaluating concatenations of words; they
follow from \requation*{inte-eval}.
\begin{align}
  \lequation{eval-left}
  \forall a\in \Ab\quantvrg v\in\Abs \quantsp
  \val{a\xmd v} ={}& a + \val{v} \xmd \base
  \displaybreak[0]\\
  \lequation{eval-righ}
  \forall a\in \Ab\quantvrg u\in\Abs \quantsp
  \val{u\xmd a} ={}& \val{u} + a \xmd \base^{\wlen{u}}
  \displaybreak[0]\\
  \lequation{eval-conc}
    \forall u,v\in\Abs \quantsp
  \val{u\xmd v} ={}& \val{u} + \val{v} \xmd \base^{\wlen{u}}
\end{align}

\medskip

In this article, we are interested in the set of the values of the words
accepted by automata over~$\Ab$.
For the sake of consistency, we will only consider automata~$\Ac$ that
\emph{accept by value} that is, such that~$\Ac$ accepts either all words of
value $k$, or none of them.
This acceptance convention is generally more practical than the other
one (accepting by representation): considered automata are usually
smaller, proofs are more elegant, and in the multidimensional settings
(which we do not consider here) it makes operations like projection
much more efficient.
In practice, it means that all automata we consider are such that the successor
by~$0$ of a final state exists and is final while the successor by~$0$ of a
non-final state is non-final, if it exists.
Thus, we may say without ambiguity that an automaton
\emph{accepts}~$S$, where~$S$ is a subset of~$\N$; indeed, it means
that~$\Ac$ accepts the language~$\cod{S}\xmd 0^*$.
%


\section{Automaton accepting an arbitrary periodic set}
\lsection{minim}

The purpose of this section is to define and study the minimal automaton
that accepts an arbitrary ultimately periodic set of nonnegative integers.
Similar results and constructions were used in the literature in other contexts,
for instance when considering automata for linear constraints
(\cf~\cite{ComoBoude1996}).
First, let us introduce some terminology and notation.

\clearpage
\begin{defi}\hfill
  \begin{subthm}
    \item A set of integers~$S\subseteq N$ is said to be \emph{purely periodic} if it may be written
          as~$S=\EpR$, for some positive integer~$p$ and~$R\subseteq\otop$.

          Moreover, we say that~$S$ is \emph{canonically written as}~$\EpR$
          if there is no integer~$p'$, ${0<p'<p}$, and~$R'\subseteq\otop[p']$ such that~$S=(R'+p'\xmd\N)$.

    \item Given two sets~$S$ and~$S'$, we denote by~$S\oplus S'$ their symmetric difference:
    an element belongs to~$S\oplus S'$ if it is in~$S$ or in~$S'$ but not in both.

    \item A set~$S\subseteq\N$ is said to be \emph{ultimately periodic} (u.p.)
          if it may be written as~$S=I\oplus S'$, where~$I$ is a finite subset
          of~$\N$ and~$S'$ is purely periodic.

          Moreover we say that~$S$ is \emph{canonically written as}~$\EpRI$
          if~$S'$ is canonically written as~$\EpR$.



    \item If~$S$ denotes \aupsn canonically written as~$S= \EpRI$, then we call
          \begin{itemizeinline}
            \item $p$ \emph{the period} of~$S$
            \item $R$ \emph{the remainder set} of~$S$
            \item $I$ \emph{the mismatch set} of~$S$
            \item $m$ \emph{the preperiod} of~$S$, where~$m=\max(I)+1$ if~$I\neq\emptyset$, and~$m=0$ otherwise.
          \end{itemizeinline}
  \end{subthm}
\end{defi}

\noindent
For instance, let us consider the set $S=\set{0,6}\cup\big(\set{4,5}+4\xmd\N\big)$.
It is canonically written as~${S=\set{1,6}\oplus\big(\set{0,1}+4\N\big)}$.
Hence, the period of~$S$ is~$4$; its remainder set is~$\{0,1\}$; its mismatch set
is~$\set{1,6}$; and its preperiod is~$7$.
Similarly, the period of the empty set (resp.~of~$\N$)
is~$1$ and its remainder set is~$\emptyset$ (resp.~$\set{0}$).

\subsection{The function \texorpdfstring{$\dX$}{Delta}}

In Section~\thesubsection{}, we take interest in the function~$\dX$ that later on will
be used as the common transition function of all minimal automata that accept \upsns.
We denote by~$\powerset(X)$ the set of the subsets of~$X$.

\begin{defi}\ldefinition{dX}
  Let~$\dX$ be the function~$\dX:(\powerset(\N)\times \Ab)\rightarrow \powerset(\N)$
  defined by:
  \begin{equation}\lequation{dX}
    \forall S\mathbin{\subseteq}\N\quantvrg\forall a\in\Ab \quantsp \dX(S,a) = \setq{n\in\N}{(nb+a)\in S}
    \eqpnt
  \end{equation}
  As usual,~$\dX$ is extended as a function~$(\powerset(\N)\times \Abs)\rightarrow \powerset(\N)$.
\end{defi}

\begin{table}[t]
  \centering
  \begin{tabular}{llll}
    $S$ & $a$ & & $\dX(S,a)$ \\
    \midrule
    $\set{0,3,4}$ & $0$ & $\longmapsto$& $\set{0,2}$ \\
    $\set{0}+2\xmd \N$ & $0$ &$\longmapsto$& $\set{0}+\N$ \\
    $\set{0}+3\xmd \N$ & $0$ &$\longmapsto$& $\emptyset+\N$\\
    $\set{0}\oplus(\set{0,1,2,4}+5\xmd \N)$ & $0$ & $\longmapsto$& $\set{0}\oplus(\set{0,1,2,3}+5\xmd \N)$  \\
    $\set{0}\oplus(\set{0,1,2,4}+5\xmd \N)$ & $1$ & $\longmapsto$& $\set{0,2,3,4}+5\xmd \N$ \\
  \end{tabular}
  \caption{A few values of function~$\dX$ in base~$2$}
  \ltable{dx}
\end{table}

\rtable{dx} gives a few instances of the function~$\dX$ in base~$2$.
Given a letter~$a$ in~$\Ab$, the function~$S\mapsto\dX(S,a)$
corresponds to reading the letter~$a$, as highlighted by the next
equation (which follows from \requation*{eval-left} and
\requation*{dX}).
\begin{equation}\lequation{succ-by-a}
  \forall S\mathbin{\subseteq}\N \quantvrg
  \forall u\in \Ab^* \quantvrg
  \forall a\in\Ab \quantsp
  \val{a\xmd u}\in S \iff \val{u}\in \dX(S,a)\eqpnt
\end{equation}

\medskip

First, let us prove that the function~$\Delta$ is stable over \upsns.

\begin{lem}\llemma{succ-of-set}
  If~$S$ denotes a set of nonnegative integers, then the following are equivalent.
  \begin{subthm}
    \item $S$ is \upf
    \item For every~$a$ in $\Ab$,~$\dX(S,a)$ is \upf
  \end{subthm}
\end{lem}
\begin{proof}
  (i) $\implies$ (ii).
  We canonically write~$S$ as~$S=\EpRI$ and we denote by~$m$ the preperiod
  of~$S$. Moreover, we write
  \begin{equation}\lequation{decr-per-preper}
    m'= \bceil{\frac{m-a}{b}} \quad\quad\text{and}\quad\quad p' = \frac{p}{\gcd(p,b)}
  \end{equation}
  Let~$n\geq m'$ be an integer.
  From \requation*{decr-per-preper}, we have~$(n\xmd b +a) \geq m$ and~$p\divides (p'b)$.
  %
  %
  The proof of the forward direction is concluded by the following equivalences:
  \begin{align*}
    n\in\dX(S,a) &{}\iff (n\xmd b + a) \in S
                 \iff (n\xmd b + a + bp') \in S
                 \\&{}\iff (n+p')\xmd b + a \in S
                 \iff (n+p')\in \dX(S,a)
  \end{align*}

  \noindent
  (ii) $\implies$ (i).
  From \requation*{dX} and the properties of Euclidean divisions, it holds that
  \begin{equation*}
    S= \bigcup_{a\in\Ab} (b\times \dX(S,a)+a) \eqpnt
  \end{equation*}
  Since multiplication, addition and finite union preserves ultimate periodicity,~$S$ is \up
\end{proof}

While showing the forward direction of \rlemma{succ-of-set}, we also showed
the following properties.

\pagebreak[2]

\begin{properties}\lproperty{decr-per-preper}
  Let~$S$  be \aupsn  and~$a$ be a letter.
  Let~$p$ be the period and~$m$ the preperiod of~$S$.
  Let~$p'$ and~$m'$ be the period and preperiod of~$\dX(S,a)$.
  \begin{subthm}
    \item \lproperty{peri-decr} $p\geq p'$
    \item \lproperty{pre-peri-decr} $m\geq m'$
    \item \lproperty{peri-not-copr}
      If $p$ is not coprime with~$b$, then $p>p'$
    \item \lproperty{pre-peri-ge-1}
      If~$m>1$, then~$m>m'$
    \item \lproperty{pre-peri-a-neq-0}
      If~$m>0$ and $a\neq0$, then~$m>m'$
  \end{subthm}
\end{properties}
%
%


\noindent
We conclude our preliminary study of~$\dX$ by two technical statements that will
be useful later on.

\begin{pty}\lproperty{dx-prod}
  Let~$I$ and~$P$ be subsets of~$\N$ such that~$I$ is finite and~$P$ is purely periodic.
  %
  For every letter~$a$ in~$\Ab$,~$\dX(I\oplus P,a)=\dX(I,a)\oplus\dX(P,a)$.
\end{pty}
\begin{proof} Let~$u$ be a word in~$\Abs$.
  \begin{align*}
     \val{u} \in \dX(I\oplus P,a) & \iff \val{au}\in I\oplus P \\
    &\iff \val{au} \in I \quad\text{or}\quad \val{au} \in P \quad\text{but not both}\\
    &\iff \val{u} \in \dX(I,a) \quad\text{or}\quad \val{u} \in \dX(P,a) \quad\text{but not both} \\
    &\iff \val{u}\in \big(\dX(I,a)\oplus\dX(P,a)\big)
    \qedhere
  \end{align*}
\end{proof}

\begin{lem}\llemma{Xc-corr}
  If~$S$ denotes \aupsn, then the following hold.
  \begin{gather}
    \lequation{Xc-corr-i} \forall u\in \Abs\quantsp \val{u}\in S \iff \dX(S,u)\mathbin{\ni} 0 \\
    \lequation{Xc-corr-ii} S = \setq{u\in\Abs}{\dX(S,u)\mathbin{\ni} 0}
  \end{gather}
\end{lem}
\begin{proof}[Sketch]
  \requation{Xc-corr-i} is shown with an induction based on \requation{succ-by-a}
  while \requation*{Xc-corr-ii} is a reformulation of \requation*{Xc-corr-i}.
\end{proof}

\subsection{The class \UP}

In the following, we manipulate sets of sets of nonnegative integers (\ie,
subsets of~$\powerset(\N)$).
For the sake of clarity, we denote such objects with a bold font.
If~$S$ and~$T$ denote two subsets of~$\N$, we say that~$T$ is
\emph{$\dX$-reachable} from~$S$ if there is a word~$u$ in~$\Abs$ such
that~$\dX(S,u)=T$.
If~$S$ is \up, \rlemma{succ-of-set} yields that~$T$ is also \up\@
Moreover, Properties~\rproperty*{peri-decr} and~\rproperty*{pre-peri-decr}
ensure that the period and preperiod of~$T$ are smaller than the ones of~$S$.
Hence, finitely many sets are $\dX$-reachable from~$S$.

\begin{defi}\ldefinition{AcS}
  Let~$S$ be \aupsn and let~$\Qf\subseteq\powerset(\N)$ be
  the set of all sets $\dX$-reachable from~$S$.
  We denote by~$\AcS$ the automaton defined by:
  \begin{equation*}
    \AcS=\aut{\Ab,\,\Qf,\,S,\,\dX_{|\Qf},\,\Ff} \eqvrg
  \end{equation*}
  where~$\dX_{|\Qf}$ is  the restriction of~$\dX$ to~$\Qf\times \Ab$
  and~$\Ff=\setq{T\in\Qf}{T\mathbin{\ni}0}$.
\end{defi}

Then, the next proposition follows directly from \rlemma{Xc-corr}.

\begin{prop}
  For every \up set~$S\subseteq\N$, the automaton~$\AcS$ is the minimal automaton
  that accepts~$S$.
\end{prop}

\begin{defi}
  We denote by \UP the class of all minimal automata that accept \upsns:
    $\text{\UP} = \setq{\AcS}{S\text{ is \aupsn}}$~.
\end{defi}

Now, let us translate \rlemma{succ-of-set} in terms
of automata for future reference.

\begin{lem}\llemma{up-sep-succ}
  Let $\Ac=\aut{\Ab,\,Q,\,i,\,\delta,F}$ be an automaton. The following are equivalent.
  \begin{subthm}
    \item $\Ac$ belongs to \UP,
    \item For every letter~$a\in\Ab$, the automaton $\Bc_a$ belongs to \UP,
          where~$\Bc_a$ is the reachable part of~$\aut{\Ab,\,Q,\,\delta(i,a),\,\delta,F}$.
  \end{subthm}
\end{lem}

\subsection{Atomic automata}

In this subsection, we take interest in some automata from \UP, that we call
\emph{atomic}.
%
%

\begin{defi}
  An automaton~$\AcS$ in~\UP is said \emph{atomic} if~$S$ is a purely periodic
  and its period is coprime with~$b$.
  For short, we say that an automaton is \emph{\UPatom} if it belongs
  to \UP and is atomic.
\end{defi}

\begin{figure}
  \begin{minipage}[b]{0.4\linewidth}\centering
    \begin{tabular}{cccc}
    $e$ &  $h_5(e,0)$ & $h_5(e,1)$ \\
    \midrule
    $0$ & $0$ & $2$ \\
    $1$ & $3$ & $0$ \\
    $2$ & $1$ & $3$ \\
    $3$ & $4$ & $1$ \\
    $4$ & $2$ & $4$
    \end{tabular}
    \renewcommand{\figurename}{Table}%
    \captionof{figure}{The function~$h_5$}%
    \ltable{h5}%
  \end{minipage}%
  \hfill%
  \begin{minipage}[b]{0.55\linewidth}\centering
    \includegraphics[scale=\AutScale]{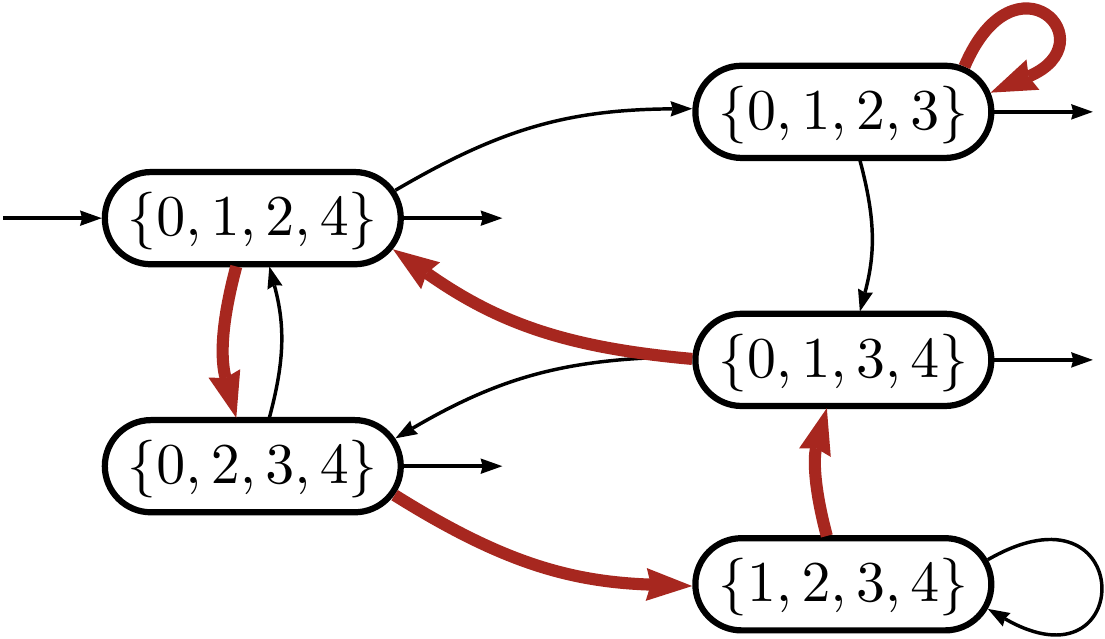}%
    \captionof{figure}{The automaton~$\AcS[(\set{0,1,2,4}+5\xmd\N)]$}%
    \lfigure{aut_2_U}%
  \end{minipage}
\end{figure}

Next, we work towards an explicit definition of \UPatom automata.
For every positive integer~$p$ coprime with~$b$, we denote by~$h_p$ the
function $(\ZZ\times\Ab)\rightarrow\ZZ$ defined by
\begin{equation}
  \forall e\in\ZZ\quantvrg\forall a\in\Ab\quantsp h_p(e,a) = (e-a)\xmd b^{\mo}\eqvrg
\end{equation}
where~$b^{\mo}$ denotes the inverse of~$b$ in~$\ZZ$.
For instance, \rtable{h5} gives the explicit definition of~$h_5$.
Let us denote by~$\powerset_{k}(\ZZ)$ the set of the subsets of~$\ZZ$ that have
cardinal~$k$.
Note that for each letter~$a$, the function~$e\mapsto h_p(e,a)$ is a
permutation of~$\ZZ$.
Hence, when~$k$ is fixed, we may lift~$h_p$ to a
function~$(\powerset_{k}(\ZZ)\times\Ab)\rightarrow\powerset_{k}(\ZZ)$
as usual:~$h_p(E,a) = \setq{h_p(e,a)}{e\in E}$.

\begin{prop}\lproposition{expl-atom}
  Let~$\EpR$ be a purely periodic subset of~$\N$.
  The automaton~$\AcS[(\EpR)]$ is isomorphic to the reachable part of
  \begin{equation*}
    \aut{\Ab,\, \powerset_{k}(\ZZ),\, R,\, h_p,\, \Ff} \quantvrg
  \end{equation*}
  where~$k=\card{R}$,~$\Ff=\setq{E \mathbin{\in} \powerset_{k}(\ZZ)}{E\ni 0}$,
  and~$R$ is naturally lifted to an element of~$\powerset_{k}(\ZZ)$.
\end{prop}
\begin{proof}
  We denote by~$\Bc$ the reachable part of the automaton described in the
  statement.
  It is a routine to check that the function
  \begin{equation*}
    \forall E\in\powerset_{k}(\ZZ)\quantsp f(E) = E+p\xmd \N
  \end{equation*}
  is an isormorphism~$\Bc\rightarrow \AcS[(\EpR)]$.
\end{proof}

%
%

%
For instance, Figures~\rfigure*{aut_2_U} and~\rfigure*{aut_2_U2} show
respectively the automata~$\AcS[(\set{0,1,2,4}+5\xmd\N)]$
and $\AcS[(\set{0,1}+5\xmd\N)]$, as defined
in~\rproposition{expl-atom} (function~$h_5$ is given in \rtable{h5}).

\begin{figure}\centering
  \includegraphics[scale=\AutScale]{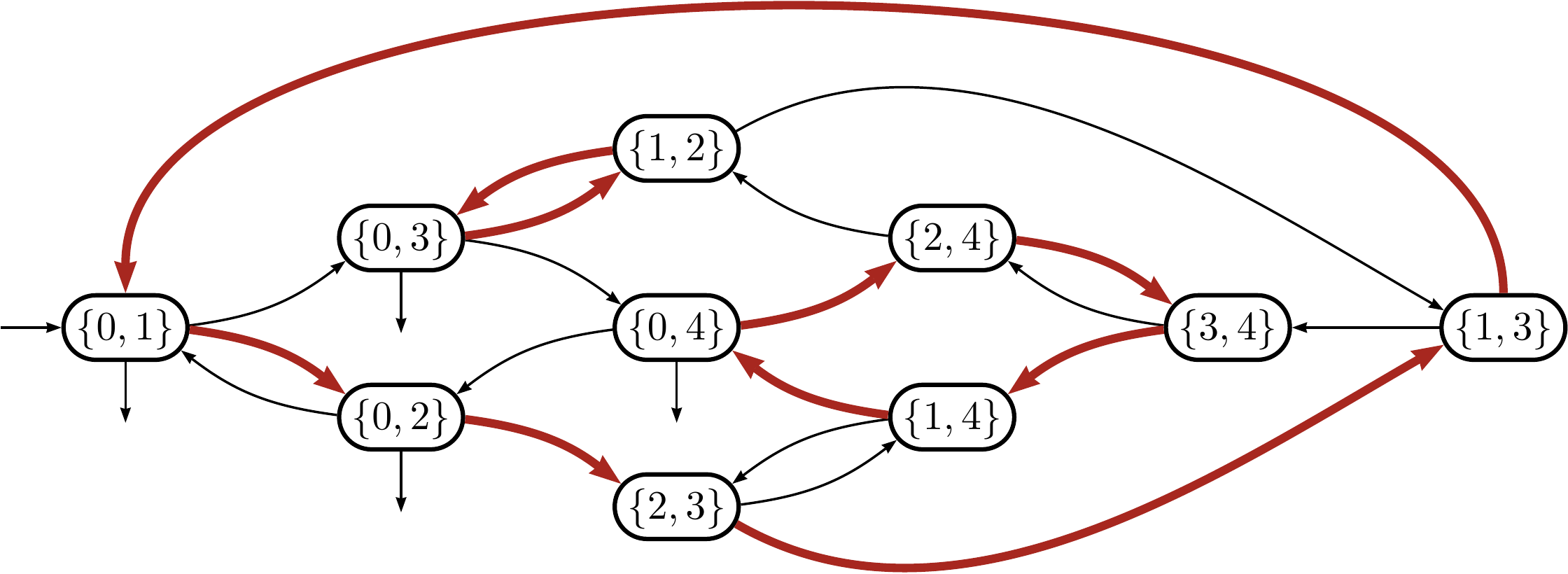}
  \captionof{figure}{The automaton~$\AcS[(\set{0,1}+5\xmd\N)]$.}
  \lfigure{aut_2_U2}
\end{figure}

\begin{lem}\llemma{atom-eq-scc}
  Let~$\Ac$ be an \upaut.
  The following are equivalent.
  \begin{subthm}
    \item $\Ac$ is atomic.
    \item $\Ac$ is a group automaton.
    \item $\Ac$ is strongly connected.
  \end{subthm}
\end{lem}
\begin{proof}
  $(\text{i})\implies(\text{ii})$.
  Since~$\Ac$ is atomic, it may be defined as per \rproposition{expl-atom}.
  As noted before, for each letter~$a$, the function~$e\mapsto h_p(e,a)$ is a permutation of~$\ZZ$, hence the function~$E\mapsto h_p(E,a)$ is a permutation of~$\powerset_k(\ZZ)$.

  \smallskip
  $(\text{ii})\implies(\text{iii})$.
  Reachable group automata are always strongly connected.

  \smallskip

  $(\text{iii})\implies(\text{i})$.
  From \rdefinition{AcS} and \rlemma{succ-of-set}, each state of~$A$ is
  a subset of~$\N$ that is \upf
  Since by hypothesis,~$\Ac$ is strongly connected,
  properties~\rproperty*{peri-decr} and~\rproperty*{pre-peri-decr}
  yield that all the states of~$\Ac$ have the same period~$p$ and
  preperiod~$m$.
  Then \rproperty{peri-not-copr} yields that~$p$ is coprime with~$b$
  and, since moreover~$\Ac$ is complete, \rproperty{pre-peri-a-neq-0}
  yields that~$m=0$.
\end{proof}

From the \rdefinition{AcS} of automata in \UP and the characterisation given in
\rlemma{atom-eq-scc}, it follows that the class of \UPatom automata is stable by
\emph{modification of the initial state}, as stated below.

\begin{pty}\lproperty{atom-chan-init}
  Let~$\Ac=\aut{\Ab,Q,i,\delta,F}$ be a \UPatom automaton.
  Then, for any~$q$ in~$Q$, the automaton~$\Bc_q=\aut{\Ab,Q,q,\delta,F}$ is \UPatom.

\end{pty}

\begin{rem}
  \rproperty{atom-chan-init} allows to say, by abuse of
  language, that some \scc is \emph{\UPatom} although it has no initial state.
\end{rem}


\section{Structural characterisation of the class \UP}
\lsection{UP}

%
%
%

The purpose of this section is to show a structural characterisation of the
class \UP (\rtheorem{UP-stru-char}).
Stating the characterisation first requires a few definitions.
%



\begin{defi}\ldefinition{embe-func}
  We say that \ascc~$C$ of an automaton~$\Ac$ is \emph{embedded} in another
  \scc~$D$ if there exists an \emph{embedding function}~$f:C\cup D\rightarrow
  D$, that is a function meeting the following.
  \begin{subthm}
    \item \ldefinition{embe-func-D-iden} For every~$s$ in~$D$, $f(s)=s$.
    \item \ldefinition{embe-func-exis-tran} For every~$s$ in~$C$ and letter~$a$, $(s\cdot a)$ exists if and only if~$(f(s)\cdot a)$ does.

    \item \ldefinition{embe-func-comm-tran} For every~$s$ in~$C\cup D$ and letter~$a$ such
    that~$(s\cdot a)\in C\cup D$, it holds that~$f(s\cdot a)=f(s)\cdot a$.

  \end{subthm}
\end{defi}

\noindent
An embedding function~$C\cup D\rightarrow D$ might be considered an
automaton ``pre-morphism'', in the sense that it satisfies
\requation*{auto-morp-trans} but not necessarily
\requation*{auto-morp-init} or \requation*{auto-morp-final}.

\begin{defi}
  We partition non-trivial \sccs in two types.
  The \emph{type two} contains the simple circuits labelled only by
  the digit~$0$, or \emph{$0$-circuits}.
  The \emph{type one} contains the other \sccs, that is each \scc with
  an internal transition labelled by a positive digit.
\end{defi}

\newcommand{\uprefi}[1]{{\rm(}{\UP}\,{\rm#1)}} 
\makeatletter
\newcommand*{\upconds}{\protect\@ifstar{\uprefi{*}\xspace}{Conditions~\protect\upconds*}}
\makeatother
\begin{thm}\ltheorem{UP-stru-char}
  An automaton~$\Ac$ belongs to~\UP if and only if the following holds,
  with~$\cond{\Ac}$ denoting the component graph of~$\Ac$.
  \begin{enumerate}[%
  leftmargin=1.3cm,%
  label=\uprefi{\arabic{*}},%
  ref=\mbox{\uprefi{\arabic{*}}},%
  start=0]
    \item\label{up.succ-0}\label{up.first}
        Each state and its successor by the digit~$0$ are both final or both non-final.

    \item\label{up.min} $\Ac$ is minimal and complete.

    \item\label{up.type1}
          Every type-one \scc is \UPatom.


    \item\label{up.type2}
          Every type-two \scc has in~$\cond{\Ac}$ exactly one \descendant,
          and that is a \scc of type one.

    \item\label{up.type2-bis}\label{up.last}
          Every type-two \scc is embedded in its \descendant in~$\cond{\Ac}$.

  \end{enumerate}
\end{thm}

\noindent
Note that Condition~\ref{up.succ-0} is not specific, it is more of a
precondition (hence its number), which ensures that the automaton accepts
by value.
Proof of \rtheorem{UP-stru-char} takes the remainder of
Section~\thesection{}.
Backward direction is shown in \rsection{up-corr} and forward direction is shown
in \rsection{up-compl}
In the following, we use \upconds* to refer to the
conditions~\ref{up.first} to~\ref{up.last}.
\begin{exa}
  \begin{figure}[t]
    \includegraphics[scale=\AutScale]{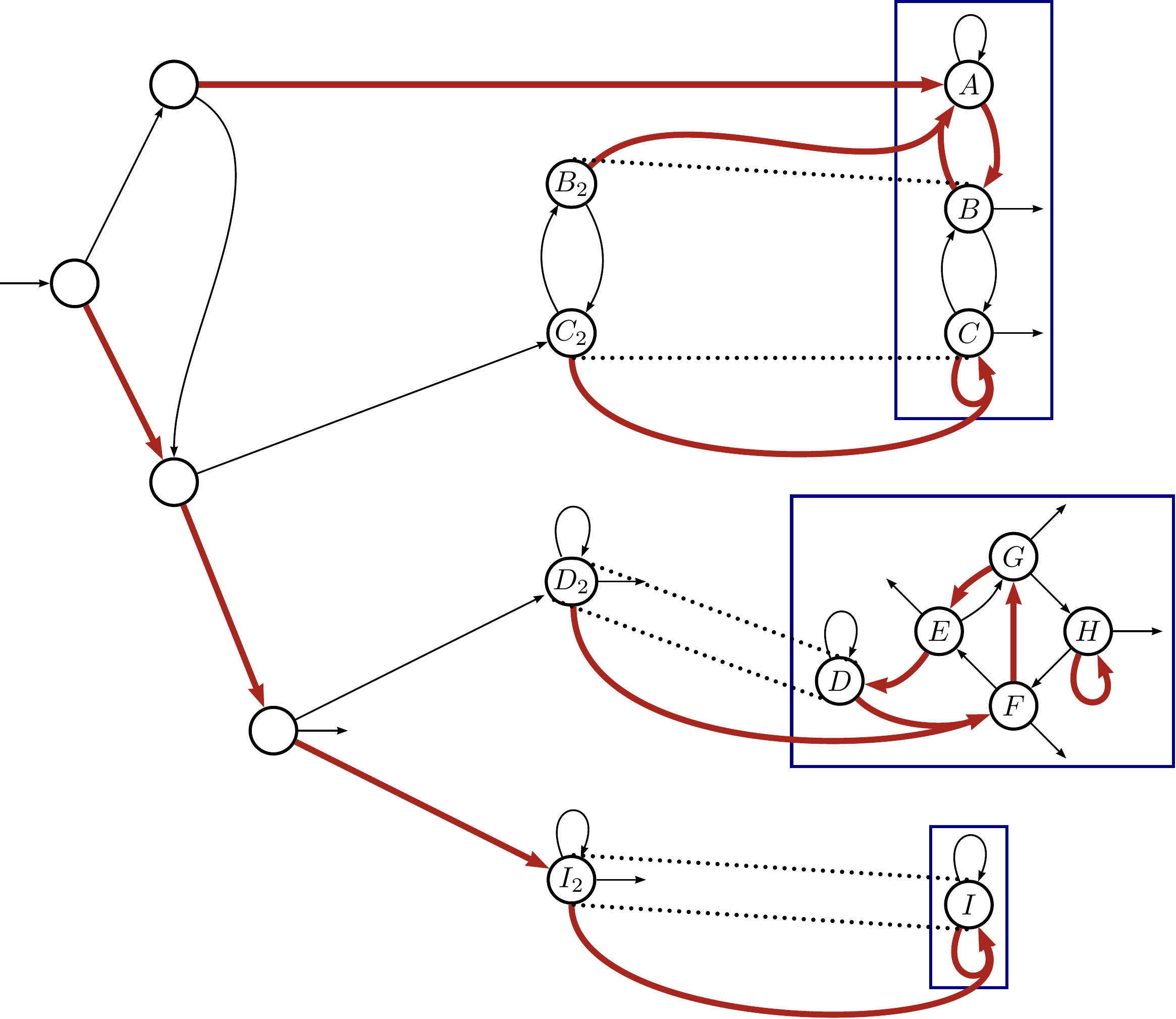}
    \caption{$\Ac_2$, an automaton that satisfies \upconds}
    \lfigure{up-example}
  \end{figure}

  \rfigure{up-example} shows an automaton~$\Ac_2$ that satisfies \upconds.
  The framed \sccs are, from top to bottom, $\AcS[S_1]$,~$\AcS[S_2]$ and~$\AcS[S_3]$
  with $S_1= (\set{1,2}+3\xmd\N)$, ${S_2=\set{0,1,2,4}+5\N}$  (\cf~\rfigure{aut_2_U})
   and~$S_3=\emptyset$.
  The three other non-trivial \sccs ($\set{B_2,C_2}$, $\set{D_2}$
  and~$\set{I_2}$) are simple 0-circuits.
  Embedding functions map each (relevant) node~$X_2$ to~$X$, with~$X$ in~$\set{B,C,D,I}$.
\end{exa}

In order to simplify the proof of both directions of \rtheorem{UP-stru-char},
we will use \rlemma{up-cond-sep-succ}, below.
It follows directly from the definition of \upconds
and states that the class of automata satisfying \upconds
possess a property much like the class \UP (\cf~\rlemma{up-sep-succ})

\begin{lem}\llemma{up-cond-sep-succ}
  Let $\Ac=\aut{\Ab,\,Q,\,i,\,\delta,F}$ be an automaton. The following are equivalent.
  \begin{subthm}
    \item $\Ac$ satisfies \upconds.
    \item For every letter~$a\in\Ab$, the automaton $\Bc_a$ satisfies \upconds,
          where~$\Bc_a$ is the reachable part of~$\aut{\Ab,\,Q,\,\delta(i,a),\,\delta,F}$.
  \end{subthm}
\end{lem}

\subsection{Backward direction of Theorem~\ref{t.UP-stru-char}}
\lsection{up-corr}

\begin{prop}\lproposition{UP-corr}
  An automaton that satisfies \upconds belongs to~\UP. 
\end{prop}
\begin{proof}
  Let~$\Ac=\aut{\Ab,\,Q,\,i,\,\delta,F}$ be an automaton that satisfies \upconds.
  Applying Lemmas~\rlemma*{up-sep-succ} and~\rlemma*{up-cond-sep-succ} allows to
  reduce the general case to the case where the initial state of~$\Ac$
  is part of a non-trivial \scc.
  If~$\Ac$ is strongly connected, it is a \UPatom automaton
  and the statement obviously holds.
  Otherwise, \upconds imply that~$\Ac$ has exactly two \sccs such that:
  \begin{itemize}
    \item the \scc containing the initial state, denoted by~$C$, is a 0-circuit;
    \item the other \scc, denoted by~$D$, is a \UPatom automaton;
    \item $C$ is embedded in~$D$, and we denote by~$f:(C\cup D)\rightarrow D$
          the embedding function.
  \end{itemize}

  \noindent
  We write~$j=f(i)$ and the automaton~$\aut{A_b,\,C,\,j,\,\delta_{|C},\,F\cap C}$
  is thus~$\AcS[\EpR]$, for some~$p\in\N$ coprime with~$b$, and~$R\subseteq\otop$.
  Let~$u$ be a word of~$\Abs$ that contains at least one non-$0$ digit.
  Since the initial \scc is a 0-circuit,~$\delta(i\cdot u)$ is a state in~$D$.
  Since~$f$ is an embedding function, it holds that~$\delta(i\cdot u)=\delta(j\cdot u)$.
  In other words, for every~$u$ such that~$\val{u}\neq0$,~$\Ac$
  accepts~$u$ if and only if~$\val{u}\in R+p\N$.
  On the other hand, since~$\Ac$ is minimal (from~\ref{up.min}),~$i$ and~$j$ must have a different
  final/non-final status.
  It follows that the set of numbers accepted by~$\Ac$ is ${\set{0}\oplus(\EpR)}$,
  hence that~$\Ac$~belongs to \UP. 
\end{proof}

\subsection{Forward direction of Theorem~\ref{t.UP-stru-char}}
\lsection{up-compl}

\begin{prop}\lproposition{up-comp}
  Every \upaut satisfies \upconds.
\end{prop}

\begin{proof}
  Let~$\AcS$ be an \upaut which we write~$\AcS=\aut{\Ab,\,\Qf,\,S,\,\dX_{|\Qf},\,\Ff}$
  (\cf~Definitions~\rdefinition*{dX} and \rdefinition*{AcS}).
  By definition,~$\AcS$ satisfies~\ref{up.succ-0} and~\ref{up.min}.

  We write~$S$ canonically as~$S=I\oplus(R+p\N)$.
  Lemmas~\rlemma*{up-sep-succ} and~\rlemma*{up-cond-sep-succ} allow to
  reduce the general case to the case where the initial state is part of a
  non-trivial \scc.
  In other words, there exists a non empty-word~$u$ such that~$\dX(S,u)=S$.

  \begin{claim}\lclaim{up-comp.copr}
    $p$ is coprime with~$b$.
  \end{claim}
  \begin{claimproof}
    For the sake of contradiction let us assume that~$p$ is not coprime with~$b$.
    We factorise~$u$ as~$u=a\xmd v$ with~$a$ in~$\Ab$.
    From \rproperty{peri-not-copr}, the smallest period of~$\dX(S,a)$ is strictly
    smaller than~$p$.
    Hence, from \rproperty{peri-decr} the smallest preperiod of~$\dX(\dX(S,a),v)=S$
    is strictly smaller than~$p$, a contradiction.
  \end{claimproof}

  \begin{claim}
    Either~$I=\set{0}$ or~$I=\emptyset$.
  \end{claim}
  \begin{claimproof}[Sketch]
    Claim \theclaim{} is proved just as \rclaim{up-comp.copr}, but using
    Properties~\rproperty*{pre-peri-ge-1} and~\rproperty*{pre-peri-decr}.
  \end{claimproof}

  The case~$I=\emptyset$ implies that~$\AcS$ is atomic; hence~$\Ac$ obviously
  satisfies \upconds.
  It remains to treat the case where~$I=\set{0}$.
  In the following, we denote by~$\Kf$ the
  set of all purely periodic subsets of~$\N$.
  We partition~$\Qf$ as~$\Cf\uplus \Df$ where:
  \begin{itemize}
    \item $\Cf$ contains every~$X$ (in~$\Qf$) that is not purely periodic;
    \item $\Df$ contains every~$X$ (in~$\Qf$) that is purely periodic.
  \end{itemize}
  Note that since~$I=\set{0}$, every~$X$ in~$\Cf$ is of the
  form~$\set{0}\oplus Y$, with~$Y\in\Kf$.

  \begin{claim} $\Cf$ is a type-two \scc.
  \end{claim}
  \begin{claimproof}
    From \rproperty{pre-peri-a-neq-0}, reading any non-0 digit from
    any state in~$\Cf$ would reach a state in~$\Df$.
    On the other hand, from \rproperty{pre-peri-decr}, no state in~$\Cf$ is
    reachable from any state in~$\Df$.
    Since~$\Ac$ is reachable, all states in~$\Cf$ are reachable from~$S$,
    hence it is necessary that~$\Cf$ is a~$0$-circuit.
  \end{claimproof}

  Let~$f$ be the function~$(\Cf\cup\Kf)\rightarrow\Kf$ defined as follows.
  For each~$X$ in~$\Cf$, there is a~$Y\in\Kf$ such that~$X=\set{0}\oplus Y$
  and we set~$f(X)=Y$.
  For each~$X$ in~$\Kf$, we set~$f(X)=X$.

  \begin{claim}\lclaim{up-comp.f-C-D}
    $f(\Cf)\subseteq \Df$
  \end{claim}
  \begin{claimproof}
    Let~$X$ be an element in~$\Cf$ and we write~$Y=f(X)$, hence~$X=\set{0}\oplus Y$.
    Note that~$\AcS[Y]$ is atomic, hence complete and strongly-connected.
    Thus, there exists a word~$w$ such that:
    \begin{enumerate*}[before={},after={.},label={\ (\roman{*})}, afterlabel={~~}, itemjoin={{; }}, itemjoin*={{; and }}]
      \item $\dX(Y,w)=Y$
      \item $w$ does not belong to~$0^*$
    \end{enumerate*}
    From (ii),~$\dX(\set{0},w) = \emptyset$, hence \rproperty{dx-prod} yields
    that~$\dX(X,w)=Y$.
    In other words, $Y$ is reachable from~$X$, hence also from~$S$ and by
    definition of~$\AcS$,~$Y\in\Qf$.
  \end{claimproof}

  \begin{claim}\lclaim{up-comp.f-embed}
    For every~$X$ in~$\Cf$ and every word~$u$ in $\Abs$,~$f(\dX(X,u))=\dX(f(X),u)$.
  \end{claim}
  \begin{claimproof}
    The whole statement reduces easily to the case where~$u$ is a letter~$a$.
    The state~$X$ may be  written has~$\set{0}\oplus Y$
    for some set~$Y$ in~$\Kf$.
    Hence, the following concludes the proof of the claim.
    \begin{align*}
       f(\dX(X,0)) ={} &f\big(\dX(\set{0}\oplus Y,\,a)\big) \\
                                        ={} & f\big(\dX(\set{0},a)\oplus \dX(Y,a)\big) & \text{(From \rproperty{dx-prod})}\\
                                        ={} & \dX(Y,a) \\
                                        ={} & \dX\big(f(\set{0}\oplus Y),\,a\big) \\
                                        ={} & \dX(f(X),\,a)
    \end{align*}
  \end{claimproof}

  We denote by~$T$ the purely periodic set~$T=f(S)$ (hence such that~$S=\set{0}\oplus T$).  Note that, from
  \rclaim{up-comp.f-C-D},~$T$ belongs to~$\Df$.

  \begin{claim}\lclaim{up-comp.aut-acS}
    The automaton~$\smash{\aut{\Ab,\,\Df,\,T,\,\dX_{|\Df},\, F\cap D}}$ is exactly~$\AcS[T]$.
  \end{claim}
  \begin{claimproof}
    Note that the states of~$\AcS$ that are reachable from~$T$ are exactly the states of~$\AcS[T]$;
    thus, it is enough to show that all states in~$\Df$ are reachable from~$T$.
    Let~$X$ be a state in~$\Df$. 
    Since~$\AcS$ is reachable, there is a word~$w$ in~$\Abs$ such that~$\dX(S,w)=X$.
    \rclaim{up-comp.f-embed} yields that~$f(\dX(S,w))=\dX(f(S),w)$.
    Since~$f(S)=T$ and~$f(\dX(S,w))=f(X)=X$, it follows that~$\dX(T,w)=X$.
  \end{claimproof}

  \rclaim{up-comp.aut-acS} implies that~$\Df$ is \ascc of
  type one and, since it is the only one, that~$\AcS$ satisfies~\ref{up.type1}.
  The only other \scc of~$\AcS$ is~$\Cf$ and it is indeed of type two and has
  exactly one descendant ($\Df$), hence~$\AcS$ satisfies~\ref{up.type2}.
  Moreover, \rclaim{up-comp.f-C-D} ensures that we may restrict~$f$ to a
  function~$(\Cf\cup \Df)\rightarrow \Df$.
  Finally, \rclaim{up-comp.f-embed} yields that f, thus restricted, indeed embeds~$\Cf$ in~$\Df$, hence
  that~$\AcS$ satisfies~\ref{up.type2-bis}.
\end{proof}


\section{Deciding membership in \UP}
\lsection{line-comp}

The goal of \rsection{line-comp} is to describe an algorithm that decides
\rproblem{line-comp} and that runs in time~$\bigo{b \xmd n}$, where~$n$ is the
number of states of the input automaton.

\begin{prob}\lproblem{line-comp}
  Given a minimal automaton~$\Ac$, does~$\Ac$ satisfy \upconds?
\end{prob}

The hard part is to obtain a linear time-complexity in the special case where the
input automaton is strongly connected.
This algorithm is developed in details in \rsection{line-pasc}.
Then, the algorithm for the general case poses no particular difficulties
and is given afterwards in \rsection{line-comp-gene}.
%

\subsection{The strongly connected case}
\lsection{line-pasc}


%
%
%
%
%

From the definition of \upconds, \ptheorem{UP-stru-char}, \rproblem{line-comp}
is the same as \rproblem{eq-can}, below, if the input automaton is strongly
connected.

\begin{prob}[\UPatom]\lproblem{eq-can}
  %
  Given as input a minimal automaton~$\Ac$,
  is~$\Ac$ \UPatom?
\end{prob}

We will see later on (\requation{gg} and \rproposition{pasc-quot-read-para})
that one can compute in linear time the only possible
period~$p$ and remainder set~$R$ that could satisfy~$\Ac=\AcS[(\EpR)]$.
In that light, \rproblem{eq-can} reduces to \rproblem{atom-build}, below.

\begin{prob}[Atomic Construction]\lproblem{atom-build}
  Given a purely periodic set~$S$, the period of which is coprime with~$b$,
  build the automaton~$\AcS$.
\end{prob}
\rproposition{expl-atom} gives an explicit construction of~$\AcS[(\EpR)]$.
However, the time complexity of this construction
is in~$\bigo{b\xmd n \times \card{R}}$, where~$n$ is the number of state
in~$\AcS$.
Since~$\card{R}$ may be up to linear in~$p$, hence in~$n$, this does not achieve
the~$\bigo{b\xmd n}$ time-complexity we require.
%


%
In the following, we use a different route to solve \rproblem{eq-can}
in linear time.
In particular, we make great use of the fact that we are provided with the
solution (the input automaton~$\Ac$).
Hence, the algorithm developped in the remainder of Section~\thesubsection{}
does \strong{not} solve \rproblem{atom-build} in linear time.

\medskip

The outline of Section~\thesubsection{} is as follows.
Let~$\EpR$ be a purely periodic set such that~$p$ is coprime with~$b$.
In \rsection{pasc-defi}, we define a special automaton~$\pascal$
(called \emph{Pascal automaton}) that accepts the purely periodic
set~$\EpR$.
\rsection{pasc-synt-mono} gives the precise structure of the transition
monoid of~$\pascal$ (which, indeed is a group).
In \rsection{pasc-quot-prop}, we consider any strict quotient%
~$\Bc$ of~$\pascal$ and study the morphism~$\phi$ that realises this
quotient.
In particular, we show that one can deduce from the structure of~$\Ac$
the values of~$p$, of~$R$ and of a
parameter~$(h,k)$ that characterises~$\phi$.
Then, \rsection{pasc-deci-Ahk} gives a way to reconstruct~$\Bc$ in
linear time when knowing~$p$, $R$, and~$(h,k)$.
Finally, we describe in \rsection{pasc-deci-algo} an algorithm to decide
whether a given automaton~$\Cc$ is the quotient of any Pascal automaton:
\begin{enumerate*}[label={(\arabic{*})},before={},after={.},afterlabel={~}, itemjoin={;\ \ }, itemjoin*={;\ and\ \ }] 
  \item compute~$p$, $R$, and~$(h,k)$ from the structure of~$\Cc$
  \item use these data to reconstruct the corresponding~$\Bc$
  \item check whether~$\Cc$ and~$\Bc$ are isomorphic
\end{enumerate*}
This algorithm also solves \rproblem{eq-can}: the input automaton~$\Ac$
of \rproblem{eq-can} is assumed to be minimal, hence~$\Ac$ is
isomorphic to~$\AcS[(\EpR)]$ if and only if~$\Ac$ is a quotient
of~$\pascal$.
%

\subsubsection{Pascal automaton: definition and elementary properties}
\lsection{pasc-defi}

Let us consider a purely periodic set, canonically written as~$\EpR$.
We moreover assume that~$p$ is coprime with~$b$.
In the following, we define a special automaton that accepts~${R+p\N}$,
called the \emph{Pascal automaton of parameter~$(\per,R)$} and denoted by~$\pascal$.
Its principle indeed goes back to the work of the philosopher
and mathematician Blaise Pascal (\cf~\cite[preface]{Saka09}).
Since \per* and \base* are coprime, \base* is an invertible element of~$\ZZ$
and there exists some (smallest) positive integer \ord* such that
\begin{equation}\lequation{psi}
  \base^{\ord} ~\equiv~ 1~~~[\per]
  \texteq{hence that}
  \forall{k\in\N}
  \quantsp
  \base^k ~\equiv~ \base^{(k\mod{\psi})}~~~[\per] \eqpnt
\end{equation}
(In other words, \ord* is the order of~$\base$ in the multiplicative group
    of the invertible elements of~$\ZZ$.)
It follows from \requation*{eval-righ} and \requation*{psi} that the
value modulo~$\per$ of a word~$u\xmd a$ can be computed using the length modulo~$\ord$
and the value modulo~$\per$ of the word~$u$:
\begin{equation}\lequation{one-more-lett}
    \forall u\in\Abs \quantvrg
    \forall a\in\Ab \quantsp
    \val{u\xmd a}\mod \per ~\equiv~ (\val{u}\mod \per) + a \xmd \base^{\wlen{u} \mod \ord}
    ~~[\per]
    \eqpnt
\end{equation}

\medskip

In the following, integers will often be used in the place of elements
that should  belong to~$\ZZ[n]$, for some~$n$; in such a case, it is understood
that the integer is lifted to its equivalence class modulo~$n$.
This typically occurs when the results of arithmetic operations are
components of states, like in \requation{pasc-tran} for instance.

\begin{defi}\ldefinition{pasc}
  The Pascal automaton of parameter~$(\per,R)$, denoted
  by~$\pascal$, is the automaton:
  \begin{equation*}
    \pascal = \aut{\Ab,~\ZZxZZ,~ (0,0) ,~ \delta,~ R\times\ZZ[\ord]}
  \end{equation*}
  where~$R$ is lifted as a subset of~$\Z/p\Z$, and the transition
  function~$\delta$ is defined by:
  \begin{equation}\lequation{pasc-tran}
    \forall (s,t) \in \ZZxZZ \quantvrg
     \forall a\in\Ab \quantsp
     \delta((s,t),a) =(s,t)\cdot a = (s+a \base^t,\, t+1)
     ~.
  \end{equation}
\end{defi}

\begin{figure}[ht]
    \centering
    \includegraphics[scale=\AutScale]{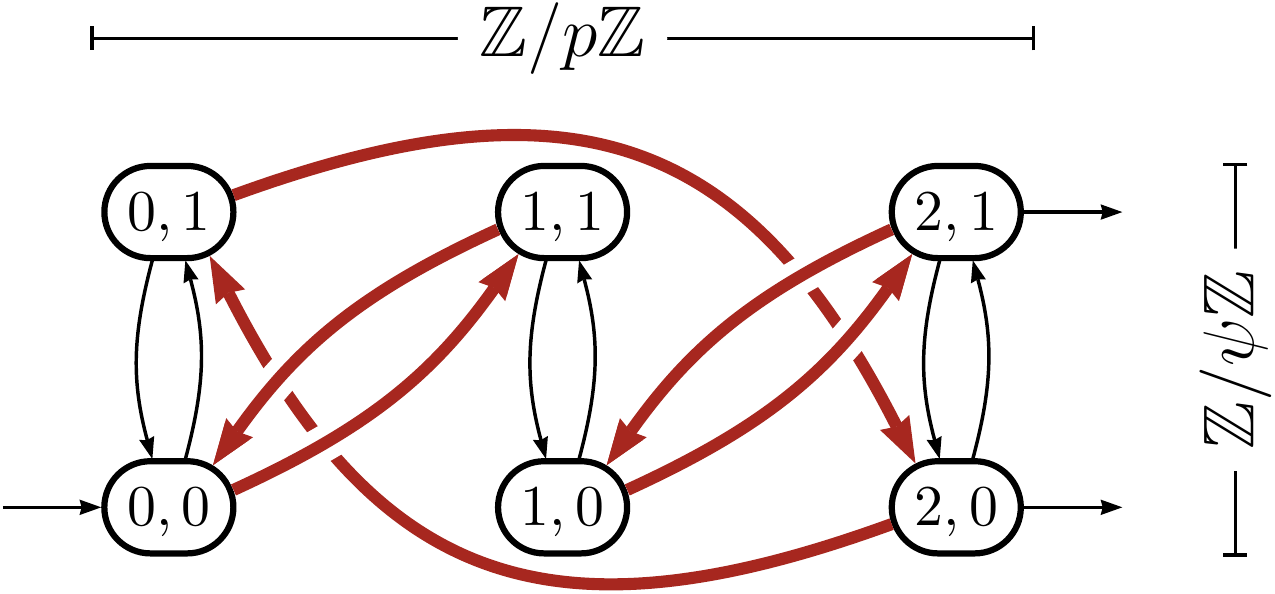}
    \caption{The Pascal automaton~$\pascal[\set{2}][3]$ in base~$2$}
    \lfigure{aut_2__2_3}
\end{figure}

\begin{exa}
  Let~$\base=2$,~$\per=3$,~$R=\set{2}$, hence~$\ord=2$.
  \rfigure{aut_2__2_3} shows the Pascal automaton~$\pascal[\set{2}][3]$.
  Recall that transition labels are omitted in figures:
  transitions labelled by~1 are drawn with a thick line
  and transitions labelled by~$0$ with a thin line.
  \rfigure{aut_2__0_7} shows~$\pascal[\set{6}][7]$; most transitions
  are dimmed for the sake of clarity.
\end{exa}

\begin{figure}[ht]
    \centering
    \includegraphics[scale=\AutScale]{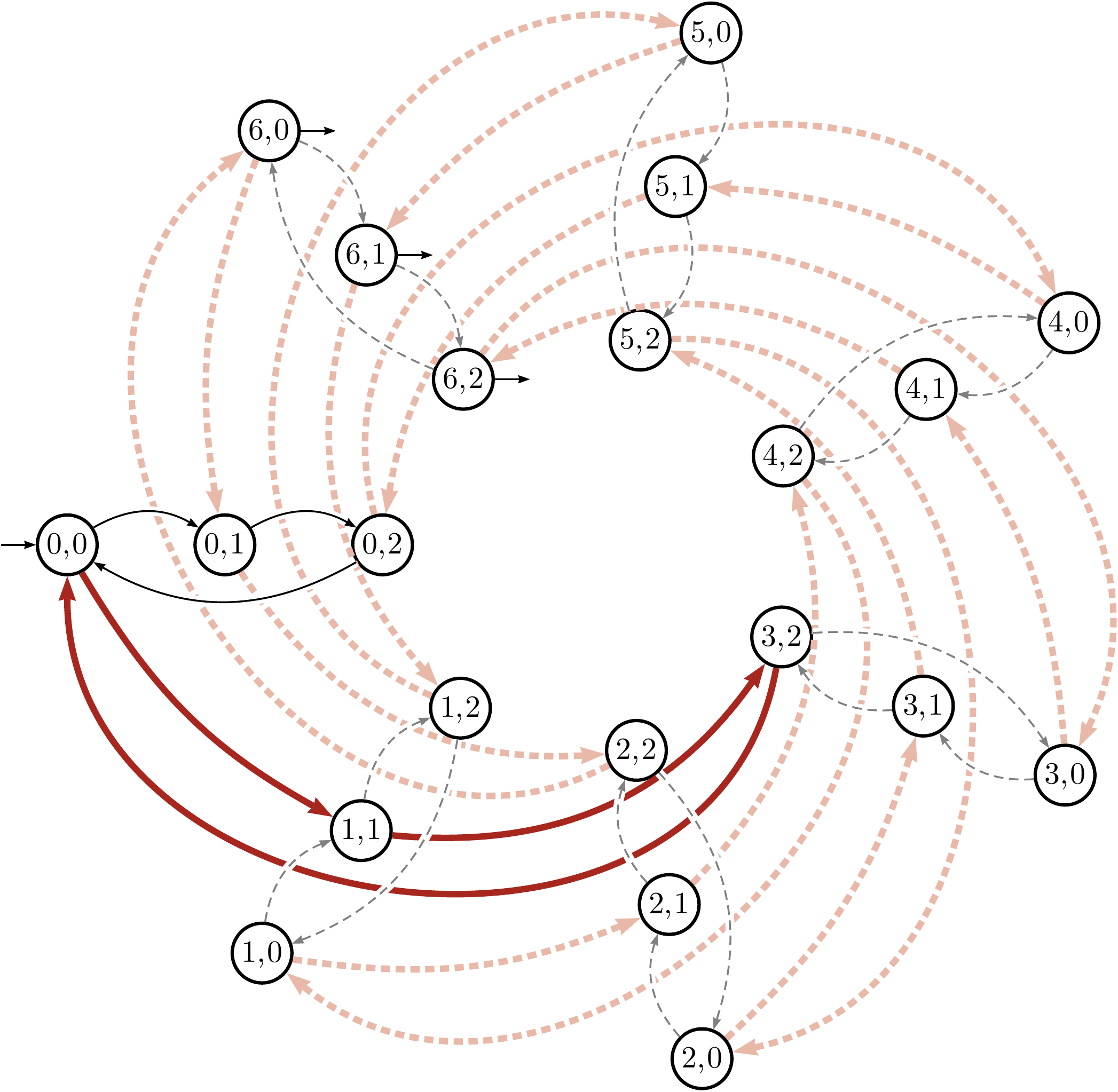}
    \caption{The Pascal automaton~$\pascal[\set{6}][7]$ in base~$2$}
    \lfigure{aut_2__0_7}
\end{figure}

Pascal automata have the expected behaviour,
as stated below.
\begin{prop}\lproposition{pasc-corr}
  The Pascal automaton~$\pascal$ accepts \EpR*.
\end{prop}
\rproposition{pasc-corr} is a direct consequence of the
\rcorollary{pasc-corr} of the next lemma, which characterises the
paths in~$\pascal$.

\begin{lem}\llemma{pasc-delt-exte}
  Let~$u$ be a word in \Abs*.
  We write~$h=\val{u}\mod\per$ and~$k=\wlen{u}\mod\ord$.
  Then, for every state~$(s,t)$ of~$\pascal$,
  \begin{equation*}
    (s,t)\cdot u=(s+h\xmd \base^t,\, t+k)
  \end{equation*}
\end{lem}
\begin{proof}
  Induction over the length of~$u$.
  The case~$\wlen{u}=0$ is trivial.
  Now we assume that~$u\neq\epsilon$.
  We write~$u=v\xmd a$ with~$a\in\Ab$ and~$v\in\Abs$.
  Moreover, we write~$h'=\val{v}\mod \per$ and~$k'=\wlen{v}\mod\ord~(=k\mo)$.
  We apply below induction hypothesis and \requation*{pasc-tran}.
  \begin{multline*}
    \big((s,t)\cdot v\xmd a\big) = (s+h'\xmd \base^t,\, t+k') \cdot a
    = (s+h'\xmd b^t+a\xmd b^{t+k'},\,t+k'+1) \\
    = (s+b^t(h'+ab^{k'}),\,t+k)
  \end{multline*}
  \requation{one-more-lett} yields the following and concludes the proof.
  \begin{equation*} \val{u}~=~\val{v\xmd a} ~\equiv~ \val{v}\mod \per + a\xmd \base^{\wlen{v}\mod \ord} \equiv~ h'+a\xmd \base^{k'}~~~[\per] \eqpnt
  \qedhere
  \end{equation*}
\end{proof}

\begin{cor}\lcorollary{pasc-corr}
  Let~$u$ be a word in \Abs*.
  The run of~$u$ in~$\pascal$ ends in the state~$\big(\val{u},\wlen{u}\big)$.
\end{cor}

%
Much like~$\AcS[(\EpR)]$, every Pascal automaton~$\pascal$ (and indeed each
quotient of~$\pascal$) is a group automaton, as stated next.
\begin{lem}\llemma{pasc-grou}
  Every Pascal automaton is a group automaton.
\end{lem}
\begin{proof}
  Let~$\pascal$ be a Pascal automaton,~$(h,k)$ a state
  of~$\pascal$, and~$a$ a letter in \Ab*.
  From \requation*{pasc-tran}, a state~$(s,t)$ is predecessor of~$(h,k)$
  by~$a$ if and only if~$s\equiv(h-a\xmd \base^{k})~[\per]$
  and~$t\equiv h\mo~[\ord]$; such a predecessor exists and is unique
  since~$\per$ is coprime with~$\base$.
\end{proof}

\begin{cor}\lcorollary{quot-pasc-grou}
  Every quotient of a Pascal automaton is a group automaton.
\end{cor}

%
The remainder of Section~\thesubsection{} is dedicated to devising an algorithm
to decide the following problem.

\begin{prob}[Quotient of a Pascal automaton]\lproblem{pasc-quot}
  Given as input an automaton~$\Ac$, is~$\Ac$ the quotient of some Pascal automaton?
\end{prob}

Note that~\rproblem{pasc-quot} is more general than \rproblem{eq-can}.
They become identical if we add in~\rproblem{pasc-quot} the extra
assumption that~$\Ac$ is minimal
%

\subsubsection{Transition monoids of Pascal automata}
\lsection{pasc-synt-mono}

For a fixed period~$\per$, and a variable remainder set~$R$, the Pascal
automata~$\pascal$ are isomorphic, aside from the final-state set.
In particular, their transition monoids are isomorphic as well.
We denote this monoid by \Gp* in the following;
it is indeed a group from \rlemma{pasc-grou}.
Let us now study the structure of this group.
We recall that~$\ord$ denotes the smallest positive integer such that~$b^\psi$
is congruent to~$1$ modulo~$p$, and that~$\ZZxZZ$ is the state set
of~$\pascal$.

\begin{prop}\lproposition{gp-isom-semi}
  The group~$\Gp$ is isomorphic to the semidirect
  product~$\ZZ[\per]\rtimes\ZZ[\ord]$.
\end{prop}

The proof of \rproposition{gp-isom-semi} requires additional definitions
and properties.
By definition of transition monoid,
\Gp* is the set of the permutations of~$\ZZxZZ$ (the state-set of~$\pascal$)
induced by words.
For every~$u$ of \Abs*, the permutation induced by~$u$, denoted by~$\tauu$,
is defined below.
\begin{equation}\lequation{defi-tauu}
  \begin{array}{rccl}
    \tauu: & \ZZxZZ & \longrightarrow & \ZZxZZ\\
           & (s,t) & \longmapsto & (s,t)\cdot u
  \end{array}
\end{equation}
The next property follows directly from \rlemma{pasc-delt-exte}.
\begin{pty}
  For every words~$u,v\in\Abs$, the permutations~$\tauu$ and~$\tauu[v]$
  are equal if and only if both~~$\val{u}\equiv\val{v}~[\per]$
  and~~$\wlen{u}\equiv\wlen{v}~[\ord]$.
\end{pty}

Hence, the group \Gp* is isomorphic to the group~$(\ZZxZZ,\gx)$; the operation~$\gx$ is defined by
\begin{equation}\lequation{defi-gx}
  (s,t) \gx (h,k) = (s+h\xmd \base^t,\, t+k) \eqpntvrg
\end{equation}
and the following function realises the isomorphism.
\begin{equation}\lequation{iso-gx-zzzz}
  \begin{array}{rccl}
    g:  & \Gp & \longrightarrow & \ZZxZZ\\
        & \tauu & \longmapsto & \bp{\val{u},\wlen{u}}=\tauu\bpp{0,0}
  \end{array}
\end{equation}
We may rephrase the same fact by linking the transition function of~$\pascal$
(\cf~\rlemma{pasc-delt-exte}) to the~$\gx$ operation.
\begin{equation}
  \forall u\in\Abs \quantvrg \forall (s,t)\in\Gp \quantsp
    \bp{(s,t)\cdot u}= (s,t) \gx (\val{u},\wlen{u}) \eqpnt
\end{equation}
\medskip
The next properties conclude the proof of \rproposition{gp-isom-semi}.
We recall that a subgroup~$H$ of a group~$G$ is \emph{normal}
if every~$x$ in~$G$ is such that~$x\xmd H\xmd x^{\mo} \subseteq H$.

\begin{properties}\hfill
  \begin{subthm}
    \item The set~$H=\ZZ\times\set{0}$ is a normal subgroup of~$\Gp$.
    \item The set~$K=\set{0}\times\ZZ[\ord]$ is a subgroup of~$\Gp$.
    \item \lproperty{gn-inte-semi-prod}
    The group~$\Gp$ is the internal semi-direct product~$H\rtimes K$.
  \end{subthm}
\end{properties}
\begin{proof}\hfill
  \begin{enumerate}[(i)]
  \item
  Let~$(h,0)$ and~$(h',0)$ be two elements of~$H$.
  From \requation*{defi-gx}, their product~$(h,0) \gx (h',0) = (h+h',0)$ is
  indeed an element of~$H$.
  Thus,~$H$ is a subgroup of~$\ZZxZZ$.
  Let~$(s,t)$ be an element in~$\ZZxZZ$.
  It follows from \requation*{defi-gx}
  that the second component of its inverse, ${(s,t)}^{\mo}$, is necessarily~$-t$
  modulo~$\ord$.
  Hence, for every element~$(h,k)$ of~$\ZZxZZ$,
  the second component of~$\bp{(s,t)\gx (h,k) \gx {(s,t)}^{\mo}}$ is equal to~$k$.
  The case~$k=0$ yields that~$H$ is normal.

  \smallskip

  \item Shown similarly from \requation*{defi-gx}.

  \smallskip

  \item
  Every element~$(h,k)$ of~$\ZZxZZ$ may be factorised
  as~$(h,0)\gx(0,k)$, hence~$H\gx K = \ZZxZZ$.
  Since moreover~$H\cap K$ contains only the neutral
  element $(0,0)$, $\ZZxZZ=H\rtimes K$.
  \qedhere
\end{enumerate}
\end{proof}

%
%
In the following, we identify \Gp* with~$\ZZxZZ$; we may then write
the \emph{permutation}~$(s,t)\in\Gp$.
(It is in fact the permutation~$\tauu$, where~$u$ is any word that
satisfies~$\val{u}\equiv s~[\per]$ and~$\wlen{u}\equiv t~[\ord]$.)

\medskip

Since it is a transition monoid, \Gp* is generated by the permutations induced
by the letters of \Ab*.
On the other hand, it is isomorphic to~$\ZZxrZZ$ hence is obviously generated by
the elements~$(0,1)$ and~$(1,0)$.
The former is the permutation induced by the digit~$0$ while the latter
is not induced by a letter, but rather by the word~$1\xmd 0^{\ord\mo}$.
We define a new letter~$g$ whose action on~$\pascal$ is defined as
the one of~$1\xmd 0^{\ord\mo}$:
\begin{equation}\lequation{defi-g-pasc}
  \forall (s,t)\in \underbrace{\ZZxZZ}_{=\,\Gp} \quantsp
  (s,t)\pathx{g}[\Ac]\underbrace{(s+\base^t,\,t)}_{\scriptstyle=\,(s,t)\,\gx\,(1,0)}
  \eqpnt
\end{equation}
The next statement follows from \requation{defi-gx}.

\begin{pty}\lproperty{pasc-a-from-g}
  For every letter~$a$ of \Ab*,
  the actions of~$a$ and of~$g^a\xmd 0$ are equal.
\end{pty}

Thus, the letter~$g$ allows to simplify~$\pascal$ into an automaton
over the alphabet~$\set{0,g}$ without losing information.
This `equivalent' automaton, denoted by~$\pascalp$, is obtained by adding
the letter~$g$ (which acts as the word~$1\xmd 0^{\ord\mo})$
and then deleting every letter~$a\in \Ab$,~$a\neq 0$.

\begin{figure}[ht!]
  \centering
  \includegraphics[scale=\AutScale]{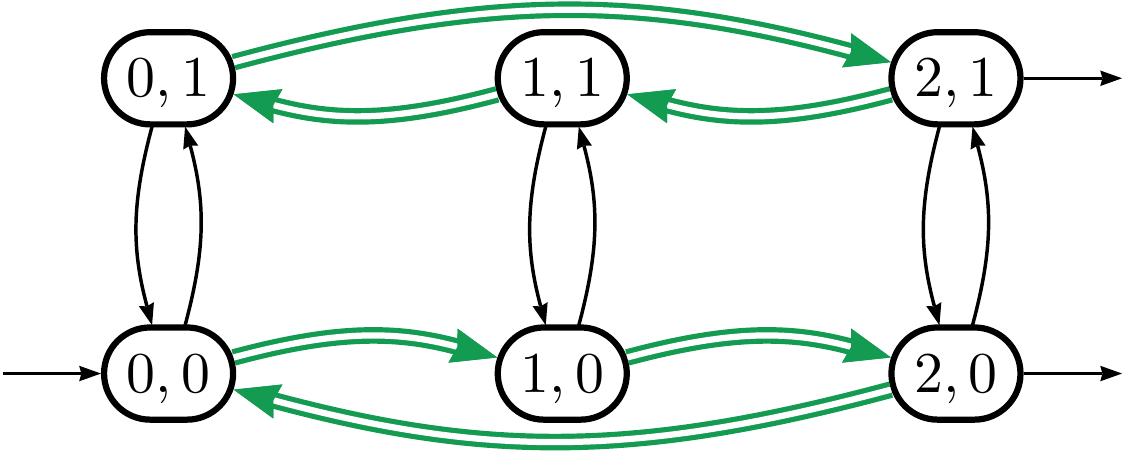}
  \caption{The simplified Pascal automaton~$\pascalp[\set{2}][3]$ in base 2}
  \lfigure{aut_2__2_3_relab}
\end{figure}

\begin{exa}
  Figures \rfigure*{aut_2__2_3_relab} and \rfigure*{aut_2__0_7_relab} show
  the automata~$\pascalp[\set{2}][3]$
  and~$\pascalp[\set{6}][7]$ respectively.
  %
  Once again, labels are omitted; transitions labelled by the digit~$0$ are
  drawn with a simple line while transitions labelled by~$g$ are drawn with a double line.

  \pagebreak[2]

  The structure of the transition monoid as a semidirect
  product is visible in \rfigure{aut_2__2_3_relab}.
  First,~$0$ induces a permutation within each column and~$g$ induces
  a permutation within each row.
  Second, the action of~$0$ is the same in each column while
  the action of~$g$ depends on the line.
  A similar observation can be made about \rfigure{aut_2__0_7_relab} by
  replacing columns and rows by spokes and concentric circles.
\end{exa}

\begin{figure}[ht!]
  \centering
  \includegraphics[scale=\AutScale]{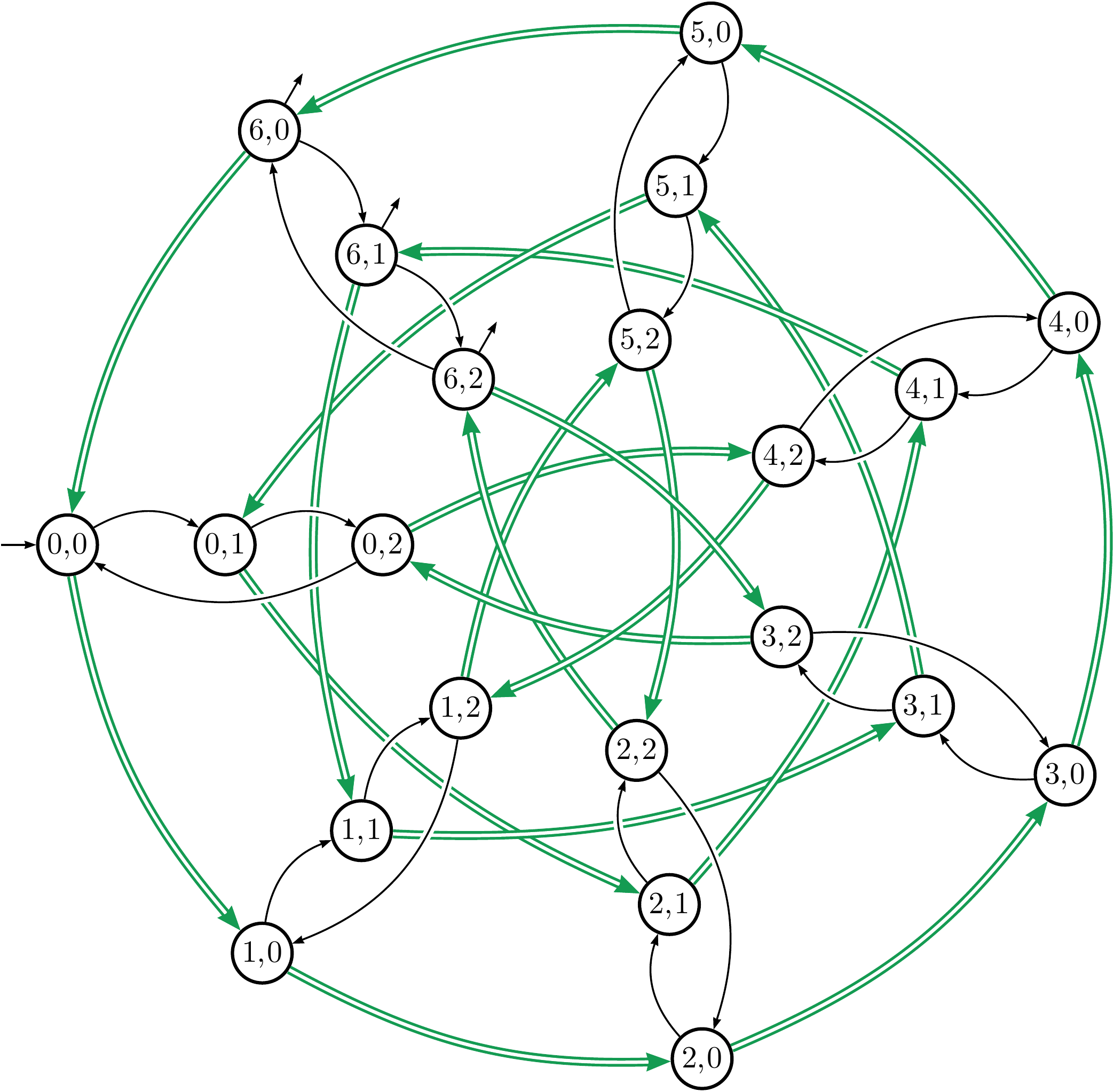}
  \caption{The simplified Pascal automaton~$\pascalp[\set{6}][7]$ in base 2}
  \lfigure{aut_2__0_7_relab}
\end{figure}

\begin{rem}
  The element~$(0, \ord\mo)$ is, in \Gp*, the inverse of~$(0,1)$.
  In \rsection{pasc-deci-algo}, we will allow to take transitions
  backward;
  the action of~$g$ is then identical to the one of the word~$1\xmd 0^{\mo}$.
  This word has the advantage to be shorter, and to be independent of~$\ord$
  (hence independent of~$\per$)
\end{rem}

\subsubsection{Properties of a quotient of Pascal automaton}
\lsection{pasc-quot-prop}

%
%
%
Here, we assume that~$\Ac$ denotes a strict
quotient of~$\pascal$, and that~$\phi$ denotes the automaton
morphism~$\pascal\rightarrow \Ac$.
Note that~$\Ac$ is a group automaton (\rcorollary{quot-pasc-grou}).
As we did for~$\pascal$ in the previous \rsection{pasc-synt-mono}, we add in~$\Ac$
transitions labelled by a new letter~$g$ whose action is the same as the
one of~$1\xmd 0^{\mo}$:
\begin{equation}\lequation{gg}
  s\pathx{g}[\Ac] s' ~\iff~ s\cdot 1 = s' \cdot 0 \eqpnt
\end{equation}
Since~$\Ac$ is a quotient of~$\pascal$, the next property follows
from \rproperty{pasc-a-from-g}.

\begin{pty}\lproperty{pasc-quot-a-from-g}
  For every letter~$a$ of \Ab*,
  the action of~$a$ in~$\Ac$ is the same as the one of the word~$g^a\xmd 0$.
\end{pty}

In the following, we give a way to compute the parameter~$(p,R)$ of the
Pascal automaton by observing the structure of~$\Ac$.  It is stated as
\rproposition{pasc-quot-read-para} after two preliminary lemmas.

\begin{lem}\llemma{phi-eq-s0-00} 
  For every~$(s,0)$ of~$\Gp$ distinct from~$(0,0)$,~$\phi\big((s,0)\big)\neq\phi\big((0,0)\big)$.
\end{lem}
\begin{proof}
  %
%
  %
  It follows from \requation*{defi-g-pasc}, which defines the transition
  labelled by~$g$ in~$\pascal$, that
  \begin{equation}\lequation{phi-eq-s0-00-i} \forall h \in\N \quantsp 
    (s,0)\pathx{g^h}[\pascal] (s+h,0) \quad\quad\text{and}\quad\quad
    (0,0)\pathx{g^h}[\pascal] (h,0)
    \eqpnt\end{equation}
  For the sake of contradiction, let us assume
  that~$\phi\big((s,0)\big)=\phi\big((0,0)\big)$;
  since~$\phi$ is an automaton morphism it follows from \requation*{phi-eq-s0-00-i}  that 
  \begin{equation*} \forall h \in \ZZ \quantsp \phi\big((s+h,0)\big)=\phi\big((h,0)\big) \eqpnt\end{equation*}
  %

  Since an automaton morphism preserves final states,
  for every~${h\in \ZZ}$,~$(h,0)$ is final if and only if~$(h+s,0)$ is final.
  From the \rdefinition{pasc} of Pascal automata (\pdefinition{pasc}), a state is final if and only if its
  first component belongs to~$R$.
  Hence,
  \begin{equation*} \forall h \in \ZZ \quantsp  h\in R \iff h+s\in R\eqpnt\end{equation*}
  In other words,~$s$ is a period of \EpR* strictly smaller than~$p$, a
  contradiction.
\end{proof}

\begin{lem}\llemma{eq-phi-st-ht}
  Let~$(s,t)$ and~$(h,k)$ be two distinct elements of~$\Gp$.
  If~$t=k$,
  then ${\phi\big((s,t)\big) \neq \phi\big((h,k)\big)}$.
\end{lem}
\begin{proof}
  Let~$(s,t)$ and~$(h,k)$ be two distinct states of~$\pascal$
  such that~$t=k$.
  We write~$u=0^{\ord-t}\xmd g^s$.
  This word labels the two following paths:
  \begin{equation*}
    (s,t) \pathx{u} (0,0) \qquad\text{and}\qquad (h,k)\pathx{u} (h-s,\, 0) \eqpnt
  \end{equation*}
  Since~$(s,t)$ and~$(h,k)$ are distinct, we necessarily have that~$(h-s)\neq 0$.
  It follows from the previous equation, that if~$\phi\big((s,t)\big)$ and
  $\phi\big((h,k)\big)$ were equal, so would be~$\phi\big((0,0)\big)$ and
  $\phi\big((h-s,0)\big)$, a contradiction to~\rlemma{phi-eq-s0-00} above. 
\end{proof}

\begin{prop}\lproposition{pasc-quot-read-para}
  Let~$\Ac$ be the quotient of a Pascal automaton~$\smash{\pascal}$.
  \begin{subthm}
    \item The circuits induced by the letter~$g$ in~$\Ac$ are all of
    length~$\per$.
    \item The word~$g^r$ is accepted by~$\Ac$ if and only if~$r$ belongs
    to~$R$.
  \end{subthm}
\end{prop}
\begin{proof}\hfill
\begin{enumerate}[(i),left=-1em]

  \item
  Let~$k$ be an element of~$\ZZ[\psi]$.
  The $g$-circuit of~$\pascal$ that contains the state~$(0,k)$ is
  \begin{equation*} (0,k)
  \pathx{g}[\pascal] (p^k, k)
  \pathx{g}[\pascal] (2\xmd p^k, k)
  \pathx{g}[\pascal] \cdots
  \pathx{g}[\pascal] ((n\mo)\xmd p^k, k)
  \pathx{g}[\pascal] (0,k)
  \eqpnt
  \end{equation*}
  The image of this circuit by~$\phi$ is:
  \begin{equation*}  \phi\big((0,k)\big)
  \pathx{g}[\pascal] \phi\big((p^k, k)\big)
  \pathx{g}[\pascal] \cdots
  \pathx{g}[\pascal]  \phi\big(((n\mo)\xmd p^k, k)\big)
  \pathx{g}[\pascal]  \phi\big((0,k)\big)
  \eqpnt
  \end{equation*}
  Since~$\phi$ is not necessarily injective, this last circuit might not be
  simple.
  In this case it would hold~$\phi((ip^k,k))=\phi((jp^k,k))$ for
  some distinct~$i,j\in\ZZ$,
  a contradiction to \rlemma{eq-phi-st-ht} above.
  Since every $g$-circuit of~$\Ac$ is necessarily
  the image of a~$g$-circuit of~$\pascalp$, item (i) holds.

  \smallskip

  \item
  The run of the word~$g^r$ ends in~$\pascalp$ the state~$(0,r)$
  which by definition is a final state if and only if~$r$ belongs to~$R$.
  Since~$\Ac$ is a quotient of~$\pascalp$, they accept the same language.
  Thus,~$g^r$  is accepted
  by~$\Ac$ if and only if it is accepted by~$\pascalp$, concluding the
  proof.
  \qedhere
\end{enumerate}
\end{proof}

\noindent
Next, we give a method to characterise the automaton
morphism~$\phi:\pascal \rightarrow \Ac$ with data observable in~$\Ac$.
Indeed, the morphism is entirely determined by the class of $\phi$-equivalence
of the state~$(0,0)$ of~$\pascal$ and in particular by the element~$(h,k)$ of
this class such that~$k$ is the smallest but still positive.
This $\phi$-equivalence class is characterised by the following lemma;
it is a consequence of the definition of the letter~$g$ in~$\pascal$.

\begin{lem}\llemma{equi-0-gsot}
  Let~$(s,t)$ be in~$\Gp$.
  The run of the word~$g^s\xmd0^t$ in~$\Ac$ reaches the initial states
  if and only if~$\phi((s,t))=\phi((0,0))$.
\end{lem}

Let~$(h,k)$ be an element of~$\Gp$.
We denote by~$\lefthk$ the permutation of~$\Gp$ induced by the multiplication
 by~$(h,k)$ \textbf{on the left} (whereas~$\tauu$ corresponds to the multiplication
 by~$(\val{u},\wlen{u})$ on the right):
\begin{equation}\lequation{defi-lefthk}
  \forall (s,t) \in \Gp \quantsp
  \lefthk\big((s,t)\big)=(h,k)\gx (s,t)= (h+s\xmd \base^k,\, k+t)
  \eqpnt
\end{equation}
We moreover write~$\lefthksub$ the permutation resulting
from the projection of~$\lefthk$ to its first component,~$\ZZ$:
\begin{equation}\lequation{defi-lefthksub}
  \forall s \in \ZZ\quantsp
  \lefthksub (s)= h+s\xmd \base^k
  \eqpnt
\end{equation}
In the following we will always consider the permutations~$\lefthk$
and~$\lefthksub$ parametrised by a special element~$(h,k)$, called by abuse
of language \emph{the smallest state~$\phi$-equivalent to~$(0,0)$}, and
defined as the unique%
\footnote{
      Uniqueness is a consequence of
      \rlemma{eq-phi-st-ht}.
    }
element satisfying the two following conditions:
\begin{itemize}
  \item $\phi((h,k))=\phi((0,0))$;
  \item every element~$(s,t)\in\Gp$ such that~$(s,t)\neq (0,0)$
  and~$\phi((s,t))=\phi((0,0))$ necessarily meets~$k<t$.
\end{itemize}
%
%
%
The next lemma follows from definitions.

\begin{lem}\llemma{lefthk-stable}
  Every~$\phi$-equivalence class is stable by the permutation~$\lefthk$
  (in~$\Gp$).
\end{lem}
\begin{proof}Let~$(s,t)$ be a state of~$\pascalp$ and~$u$
  a word such that~$(\val{u},\wlen{u})=(s,t)$.
  \begin{align*}
    \phi\big(\lefthk((s,t))\big)
        &= \phi\big((h,k) \gx (s,t)\big) \\
        &= \phi((h,k)) \ap u \\
        &= \phi((0,0)) \ap u \\
        &= \phi\big((0,0)\gx (s,t)\big) \\
        &= \phi((s,t))
    \tag*{\qedhere}
  \end{align*}
\end{proof}

\begin{rem}
  In~\cite{Mars16}, a statement stronger than \rlemma{lefthk-stable} is
  shown: the $\phi$-equivalence classes are in fact the orbits of~$\lefthk$.
\end{rem}


\subsubsection{Construction of the quotient}
\lsection{pasc-deci-Ahk}

We keep here the settings of \rsection{pasc-quot-prop}.
Namely,~$p$ denotes a positive integer coprime with~$b$,~$\Ac$ denotes
a strict quotient of~$\pascal$,~$\phi$ denotes the automaton
morphism~$\pascal\rightarrow \Ac$ and~$(h,k)$ denotes the smallest
state of~$\pascal$ that is~$\phi$-equivalent to~$(0,0)$.
The purpose of this section is to show that, given as input~$p,R,h$ and~$k$,
one can build$\Ac$ in linear time (with respect to the size of~$\Ac$).

\begin{defi}\ldefinition{Ahk}
  We denote by~$\Ahk$ the automaton
  \begin{equation*} \Ahk= \aut{\set{0,g},\, \Qhk,\, \dhk,\, (0,0),\, \Fhk } \eqvrg\end{equation*}
  where the state set is~$\Qhk=\ZZ\times\ZZ[k]$ (mind that the second operand of
  the Cartesian product is~$\ZZ[\underline{k}]$ and not~$\ZZ[\underline{\psi}]$);
  the final-state set is~$\Qhk=R\times\ZZ[k]$ (idem); the outgoing transitions
  of every state~$(s,t)\in\Q$ are defined as follows.
  {\arraycolsep=0pt\begin{equation*}
  \begin{array}{@{}r@{}l@{}}
    (s,t)\cdot 0 ~={~}&
      {
      \begin{cases}
          \; (s, t+1) & \text{if~~}  t < (k\mo) \\[.8ex]
          \;  \lefthk^{~~~\,\mo}\big((s,t+1)\big) = \displaystyle\left(\frac{s-h}{\base^{k}},\,0\right) & \text{if~~}t = (k\mo) \\
      \end{cases} }
      \\[5ex]
    (s,t)\cdot g ~={~}& (s+\base^{t},t)
  \end{array}
\end{equation*}}
\end{defi}

In the remainder of Section~\thesubsubsection{}, we show that the
automaton~$\Ahk$ is isomorphic to~$\Ac$ (\rtheorem{Ahk=A}).
The proof of this statement needs preliminary results.

\begin{lem}\llemma{phi-succ-pasc-ahk}
  Let~$(s,t)$ be an element of~$\Qhk$
  (hence both a state of~$\Ahk$ and of~$\pascal$).
  Let~$x$ be a letter of~$\set{0,g}$.
  We let~$(s',t')$ and~$(s'',t'')$ denote the successors
  of~$(s,t)$ by~$x$, respectively in~$\Ahk$ and in~$\pascal$.
  Then, as states of~$\pascal$,~$(s',t')$ and~$(s'',t'')$ are~$\phi$-equivalent.
\end{lem}
\begin{proof}
  From the definitions of~$\Ahk$ and~$\pascal$, the only case where~$(s',t')$
  and~$(s'',t'')$ are not equal happens when~$a=0$ and~$t=k\mo$.
  In this case however, we have~$\lefthk((s'',t''))=(s',t')$.
  Applying \rlemma{lefthk-stable} concludes the proof.
\end{proof}

\begin{lem}\llemma{phi-eq-cap-qhk}
  Every~$\phi$-equivalence class contains exactly one state of~$\Qhk$.
\end{lem}
\begin{proof}
  Existence.
  We denote by~$C$ any~$\phi$-equivalence class and~$(s,t)$ its smallest
  element (when ordered by second component); \rlemma{eq-phi-st-ht}
  ensures that~$(s,t)$ is well defined.
  If~$t\geq k$, then~${\lefthk}^{\!\mo}((s,t))$ is equal to~$(s',\,t-k)$ for
  some~$s'$ and it holds that~$0\leq t-k<t$.
  From \rlemma{lefthk-stable},~$(s',\,t-k)$ is
  moreover~$\phi$-equivalent to~$(s,t)$, a contradiction to the choice
  of~$(s,t)$.
  Hence~$t<k$ and~$(s,t)\in\Qhk$.

  \smallskip

  Uniqueness. Ab Absurdo.
  Let~$(s,t)$ and~$(s',t')$ two distinct and~$\phi$-equivalent states of~$\pascal$
  such that~$0\leq t,\,t' <k$.
  From \rlemma{eq-phi-st-ht},~$t$ and~$t'$ are not equal; we assume without loss
  of generality that~$t<t'$, hence it holds that~$0< t'-t<k$.
  The state~$(s',t')\gx {(s,t)}^{\mo}$ is $\phi$-equivalent to~$(0,0)$
  and equal to~$(s'',\,t'-t)$ for some~$s''$, a contradiction to the definition of~$(h,k)$
  as the smallest state~$\phi$-equivalent to~$(0,0)$.
\end{proof}

Now, we establish that $\Ahk$ is isomorphic to~$\Ac$.

\begin{thm}\ltheorem{Ahk=A}
  Let~$\pascal$ be a Pascal automaton and~$\Ac$ a non-trivial
  quotient of~$\pascal$.
  We write~$\phi$ the automaton morphism~$\pascal\rightarrow\Ac$.
  Among the states $\phi$-equivalent but not equal to~$(0,0)$,
  we denote by~$(h,k)$ the state with the smallest second component.
  Then, the automaton~$\Ac$ is isomorphic to~$\Ahk$.
\end{thm}
\begin{proof}
  We define the function~$\xi$.
  \begin{equation*}\begin{array}{lccl}
    \xi: & Q_\Ac &\longrightarrow& \Qhk \\
         & q     &\longmapsto&     \text{the unique state of~}\phi^{\mo}(q)\cap \Qhk
  \end{array}\end{equation*}
  \rlemma{phi-eq-cap-qhk} yields that~$\xi$ is well defined.
  Since the inverse images by~$\phi$ of states of~$\Ac$ are disjoint,~$\xi$ is injective.
  It is also surjective since every state~$(s,t)$ of~$\Qhk$ is the image  by~$\xi$ of~$\phi((s,t))$.
  It remains to show that~$\xi$ is an automaton morphism~$\Ac\rightarrow \Ahk$.
  The state~$(0,0)$ is necessarily mapped by~$\phi$ to~$i_\Ac$, the initial state of~$\Ac$,
  and belongs to~$\Qhk$ hence~$\xi(i_\Ac)=(0,0)$ which is the initial state of~$\Ahk$.
  Similarly,~$\phi$ preserves the final and non-final status of states hence so does~$\xi$.
  Finally,  let~$q\pathx{a} q'$ be a transition of~$\Ac$
  and let us show that~$\xi(q) \pathx{a} \xi (q')$ in~$\Ahk$.
  We denote by~$(s',t')$ and~$(s'',t'')$ the
  successors of~$\xi(q)$ by~$x$ in~$\Ahk$ and~$\pascal$, respectively.
  Since~$\xi(q)$ belongs to~$\phi^{\mo}(q)$ and since~$\phi$ is a
  morphism~$(s'',t'')$ belongs to~$\phi^{\mo}(q')$.
  Then, \rlemma{phi-succ-pasc-ahk} implies that~$(s',t')$ belongs
  to~$\phi^{\mo}(q')$ as well.
  Since~$(s',t')$ also belongs to~$\Qhk$, it holds that~$\xi(q')=(s',t')$.
\end{proof}

\subsubsection{Decision algorithm}
\lsection{pasc-deci-algo}

Let~$\Ac=\aut{Q,\Ab,\delta,i,T}$ be an automaton fixed in the following.
We will describe here an algorithm to decide whether~$\Ac$ is the quotient
of a Pascal automaton.

%
\paragraph{\trianglebullet~Step 0 (Requirements)}%
Every quotient of a Pascal automaton is necessarily a group automaton
(\rcorollary{quot-pasc-grou}) and necessarily accepts by value.
It may be verified in linear time whether~$\Ac$ satisfies these two conditions.
If it does not, reject~$\Ac$.
Moreover, we allow to take transitions (labelled by~0)
backwards; computing these transitions may be done in one traversal of~$\Ac$.
\paragraph{\trianglebullet~Step 1 (Simplification)}%
Let~$B$ be the alphabet~$\set{0,g}$.
Let us compute an automaton~$\Ac'$ over~$B$.
First, the automaton~${\Ac=\aut{Q,\Ab,\delta,i,F}}$, whose alphabet is~$\Ab$,
is transformed in the automaton~${\Bc=\aut{Q,\Ab\cup B, \delta', i, F}}$,
by adding transitions labelled by~$g$: the transition~$s\pathx{g} s'$ is added in~$\Bc$
if and only if~$s\pathx{1\xmd0^{\mo}} s'$ exists in~$\Ac$.
Second, we ensure that no information is lost in the simplification process.
From \rproperty{pasc-quot-a-from-g}, if the automaton~$\Ac$ is the quotient
of a Pascal automaton, the following equation necessarily
holds (if it does not, reject~$\Ac$): %
\begin{equation}\lequation{veri-tran-cons}
  \forall s\in Q \quantvrg \forall a \in\Ab \quantsp s\cdot a =s\cdot (g^a\xmd 0)~~~\text{in automaton }\Bc
\end{equation}
Verifying that this equation is satisfied requires to run one test for every
letter~$a$ and every state~$s$, that is one test for each transition of~$\Ac$.
It is then sufficient that each test is executed in constant time in
order for the general verification of \requation*{veri-tran-cons} to be run
in linear time.
Keeping intermediary results allows to comply to this condition.
Third, we delete from~$\Bc$ the transitions labelled by digits other
than~$0$ or~$g$ and denote the result by~$\Ac'$.

\begin{figure}[p]
  \centering%
  \includegraphics[scale=\AutScale]{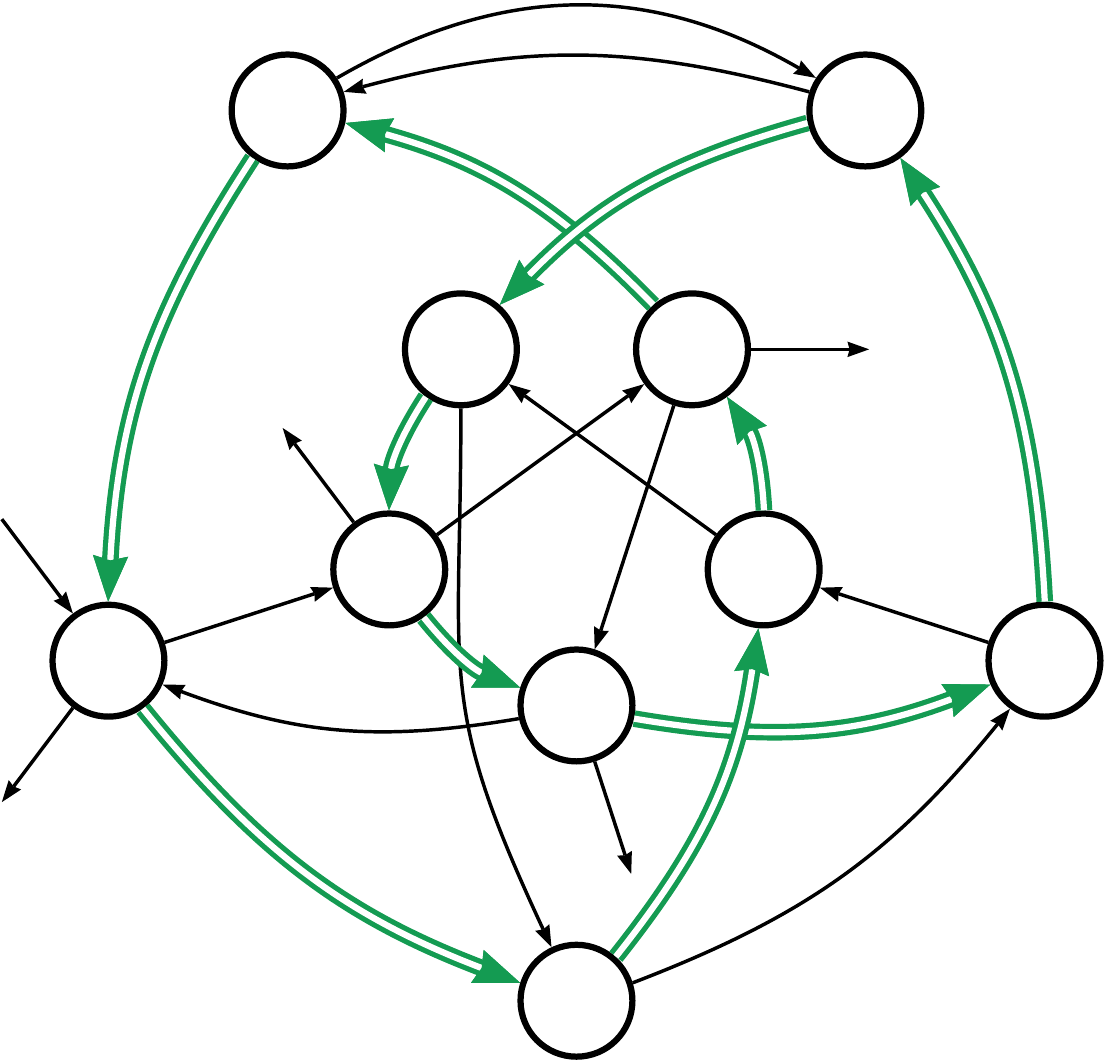}
  \caption{The simplified automaton~${\Ac_3\!}'$}
  \lfigure{pascquot_exam_0}
\end{figure}

\begin{figure}[p]

  \hspace*{-.05\linewidth}\begin{minipage}{.55\linewidth}%
    \centering
    \includegraphics[scale=\AutScale]{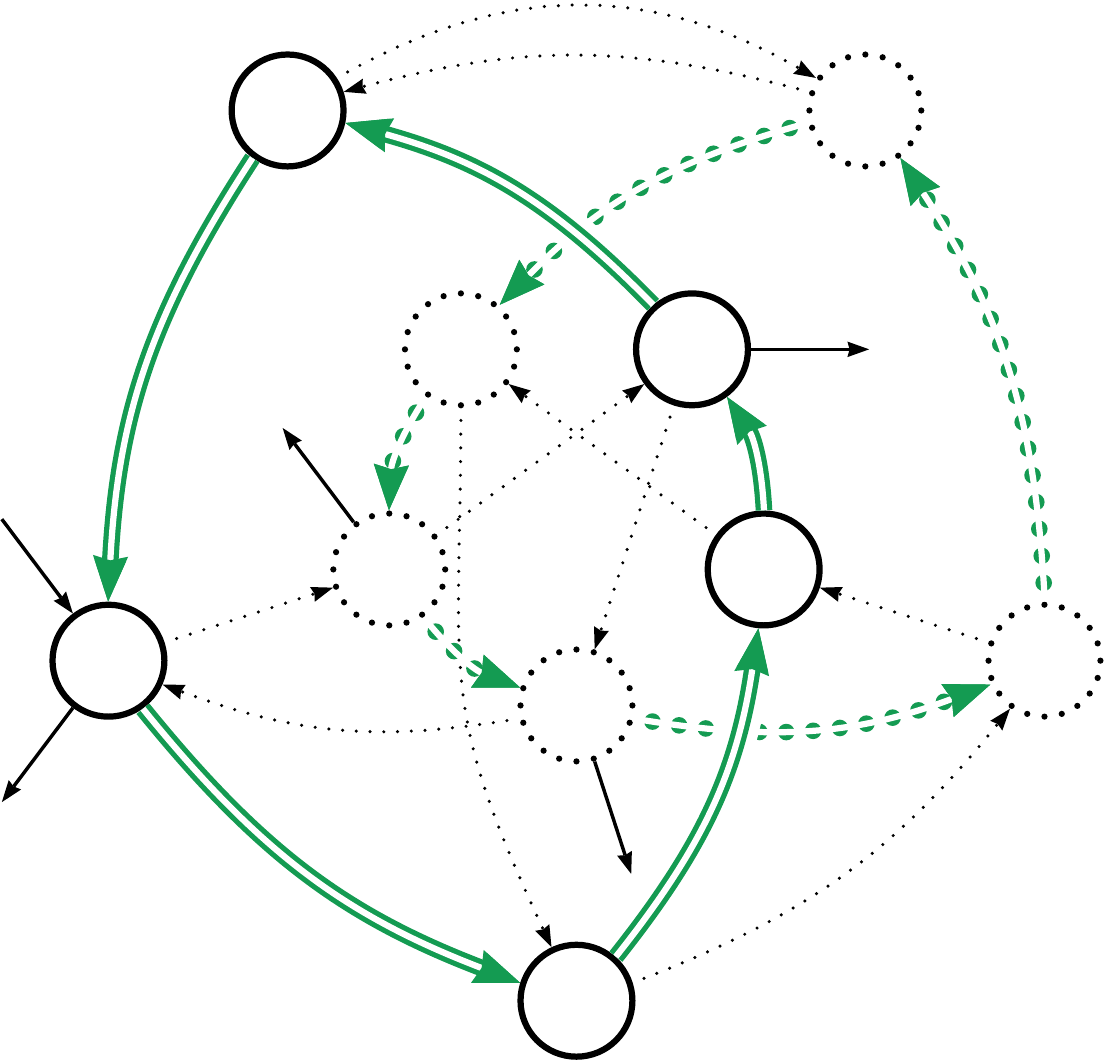}%
    \captionof{figure}{The~$g$-circuit in~${\Ac_3\!}'$ containing the initial state}
    \lfigure{pascquot_exam_1}
  \end{minipage}\hspace*{-.05\linewidth}\hfill%
  \begin{minipage}{.45\linewidth}
    \includegraphics[scale=\AutScale]{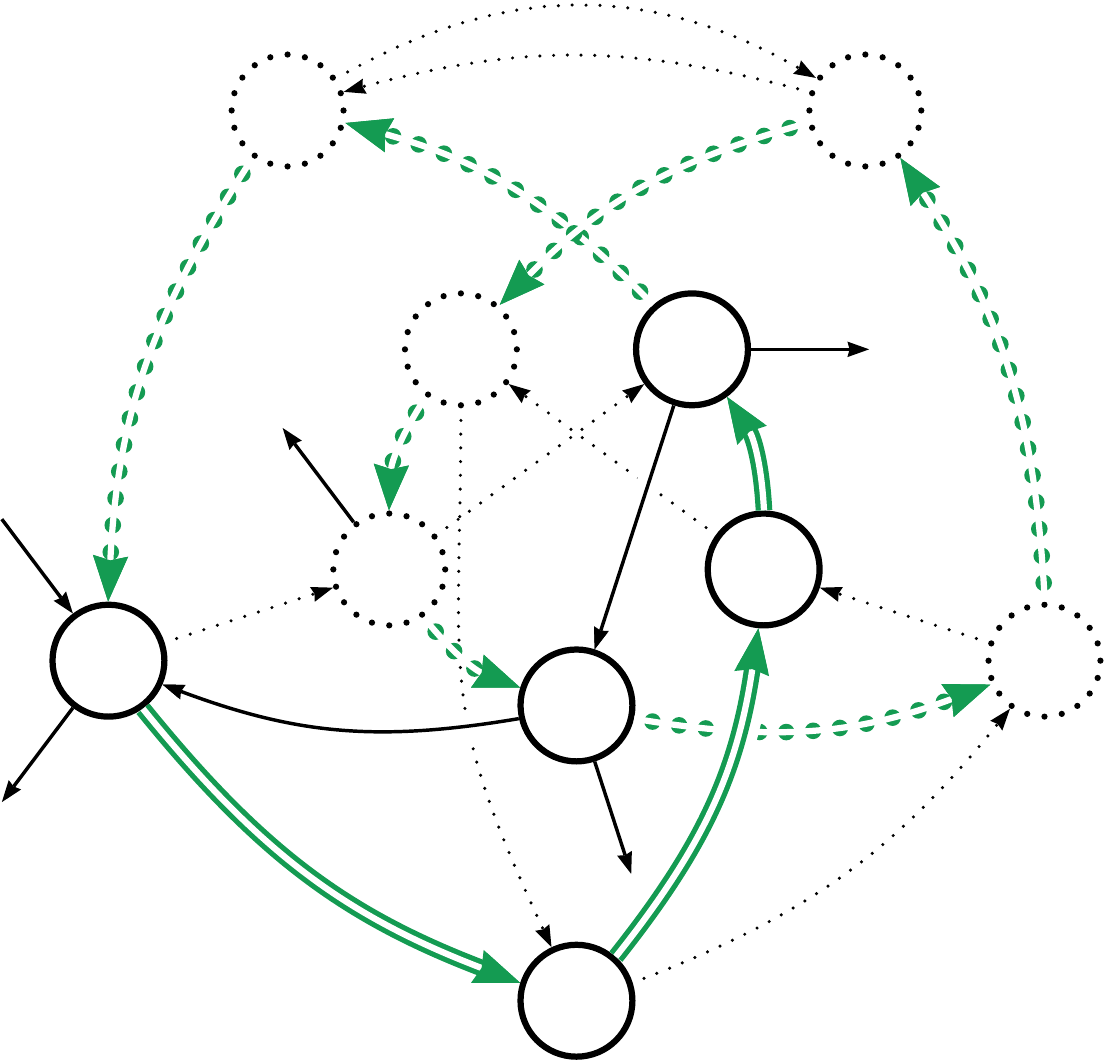}
    \captionof{figure}{The smallest mixed circuit in~${\Ac_3\!}'$}
    \lfigure{pascquot_exam_2}
  \end{minipage}\
  \hfill

\end{figure}

\begin{table}[p] \centering
    \tabcolsep=1ex%
    \begin{tabular}{|lrcl|} 
    \hline 
      \Tstrut{Base}    & $\base$ &=& 3 \\
      Period & $\per$ &=&5 \\
      Remainder set & $R$ &=&$\set{0,3}$ \\
      Parameter of the quotient  & $(h,k)$ &=&$(3,2)$ \\
      \Bstrut{Order of~$\base$ in the group~$(\ZZ[\per],\times)$} & $\psi$ &= & 4 \Bstrut{} \\
    \hline 
    \end{tabular}
    \caption{Summary of the parameters}
  \ltable{pascquot_param}
\end{table}

\begin{table}[p]\centering
    \begin{tabular}{r|rcl|l} 
      \cline{2-4}
      \Tstrut(i) & $(s,0)\cdot 0$ &= & $(s,1)$ &\\
      \Bstrut(ii) & $(s,1)\cdot 0$ &= & $(4\xmd s - 2,1)$ & \;{=} $(\frac{s-h}{p^k},0)$ \\
      \cline{2-4}
      \Tstrut(iii) & $(s,0)\cdot g$ &= & $(s+1,0)$ & \;{=} $(s+p^0,0)$ \\
      \Bstrut(iv) & $(s,1)\cdot g$ &= & $(s+3,1)$ & \;{=} $(s+p^1,1)$ \\
      \cline{2-4}
    \end{tabular}
  \caption{Transition function of~$\Ahk$}
  \ltable{tran-tab}
\end{table}

\begin{runex}
We consider an automaton~$\Ac_3$ over the alphabet~$A_3=\set{0,1,2}$, hence
the base in~$b=3$.
\rfigure{pascquot_exam_0} shows the simplified
automaton~${\Ac_3\!}'$.
(We did not include a representation of~$\Ac_3$ because it has 30 transitions.)
\end{runex}

\paragraph{\trianglebullet~Step 2 (Analysis)}
For the whole step 2, we assume that~$\Ac'$
is the quotient of a Pascal automaton~$\pascalp$, in order to
compute~$\per$,~$R$,~$\phi$ and~$(h,k)$.
(If it is not the case, these parameters have no meaning and~$\Ac$ will be
rejected during Step~3.)
We first use \rproposition{pasc-quot-read-para} to compute~$\per$ and~$R$:
\begin{itemize}
  \item $\per$ is the length of the~$g$-circuit containing the initial state;
  \item $R$ is the set of the exponents~$r$ such that~$g^{r}$ is accepted
  by~$\Ac'$.
\end{itemize}
The order~$\ord$ of~$\per$ in~$(\ZZ,{\times})$ is computed in the usual way.
The parameter~$(h,k)$ of the quotient is computed thanks to \rlemma{equi-0-gsot}:
we look for the \emph{mixed circuit}~$g^s\xmd 0^t$ with the smallest positive~$t$;
then we write~$(h,k)=(s,t)$.

\begin{runex}
\rfigure{pascquot_exam_1} highlights the~$g$-circuit containing the
initial state.
It has length 5 (as have all other~$g$-circuits), hence~$\per=5$ and final
states are at index~0 and~3, hence~$R=\set{0,3}$.
\rfigure{pascquot_exam_2} shows the mixed circuit with
the smallest number of~$0$'s (and in this case it is the only one).
Since it is labelled by the word~$g^3\xmd 0^2$, the parameter of the
quotient is~$(h,k)=(3,2)$.
\rtable{pascquot_param}  sums up all relevant parameters.
\end{runex}

\begin{figure}[p]%
  \vspace{-1em}
  \centering
  \begin{subfigure}[t]{0.41\textwidth}
    \centering
    \includegraphics[scale=\AutScale]{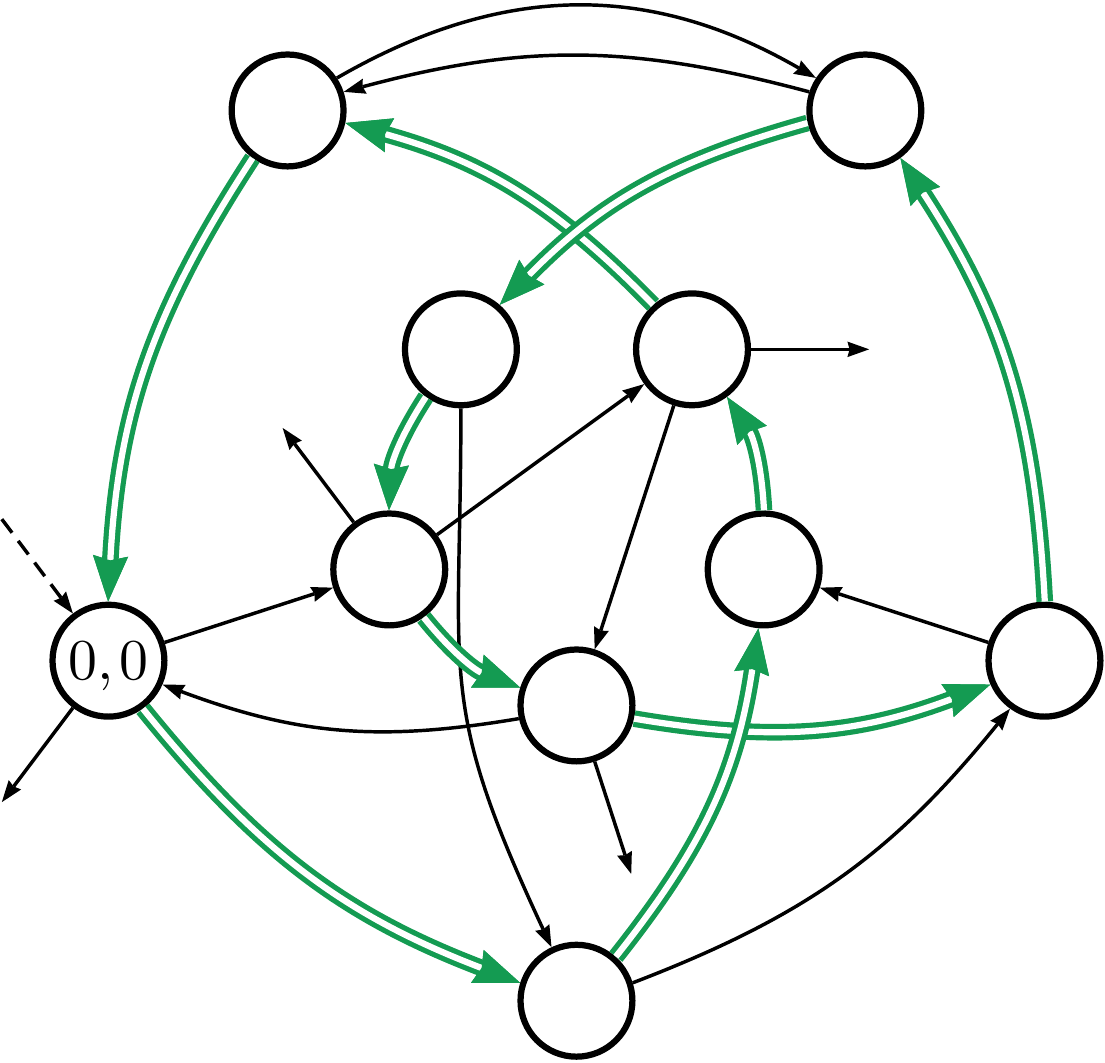}
    \caption{The initial state is coloured by~$(0,0)$}
  \end{subfigure}\hspace{15mm}
  \begin{subfigure}[t]{0.41\textwidth}
    \centering
    \includegraphics[scale=\AutScale]{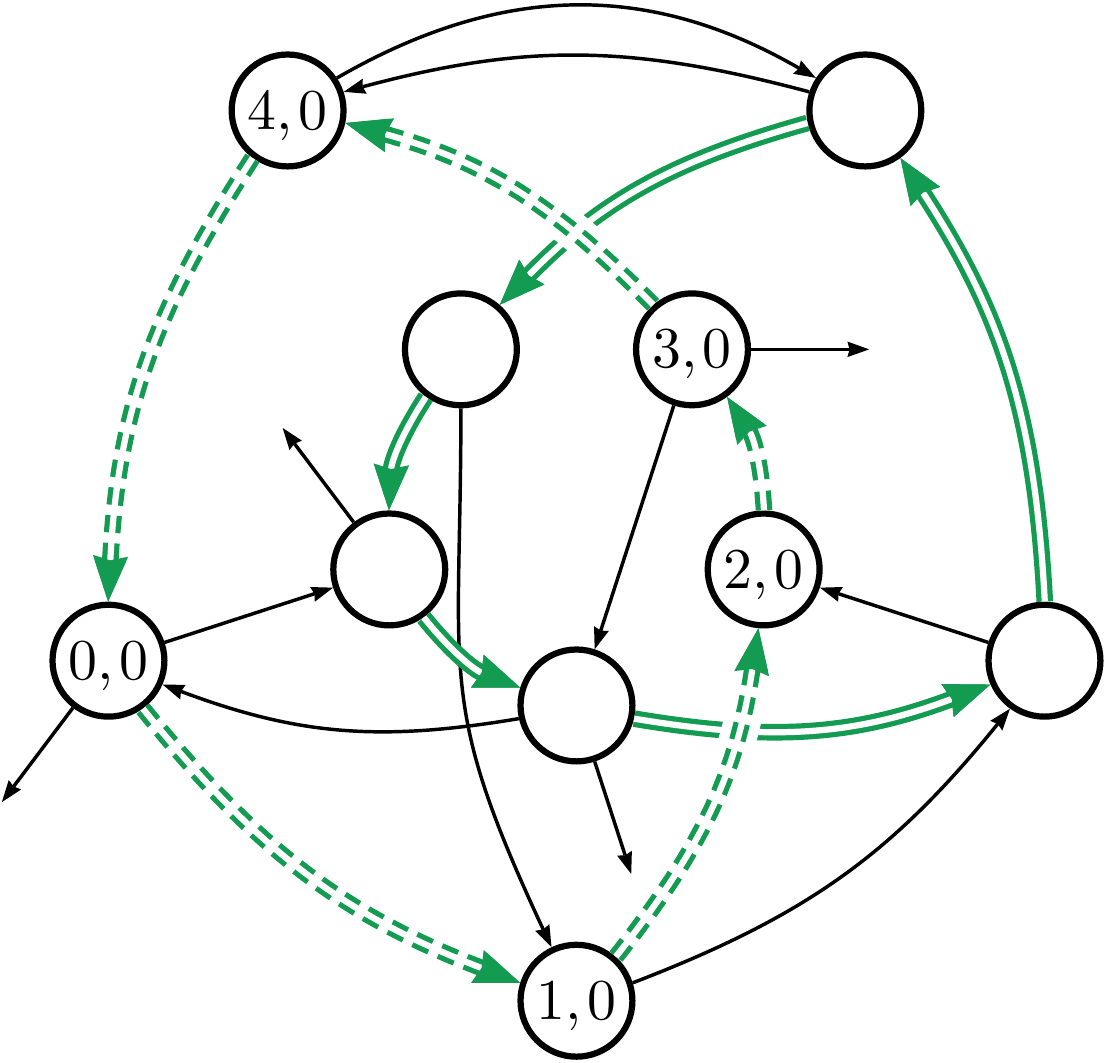}
    \caption{Applying rule (iii) of \rtable{tran-tab}: \\$(s,0)\pathx{g} (s+1,0)$}
  \end{subfigure}

  \vspace{0.5em}
  \begin{subfigure}[t]{0.41\textwidth}
    \centering
    \includegraphics[scale=\AutScale]{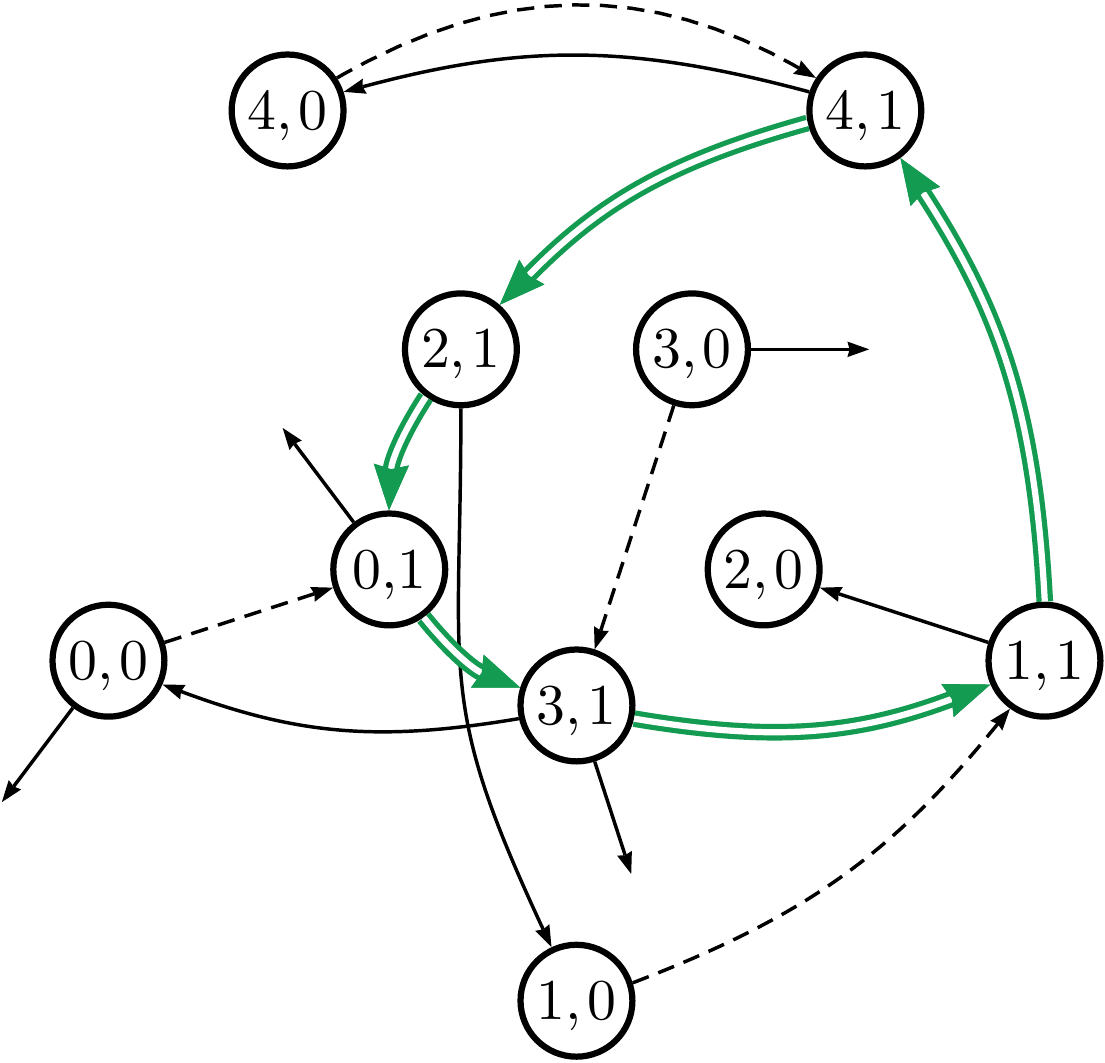}
    \caption{Applying rule (i) of \rtable{tran-tab}: \\$(s,0)\pathx{g} (s,1)$}
  \end{subfigure}\hspace{15mm}
  \begin{subfigure}[t]{0.41\textwidth}
    \centering
    \includegraphics[scale=\AutScale]{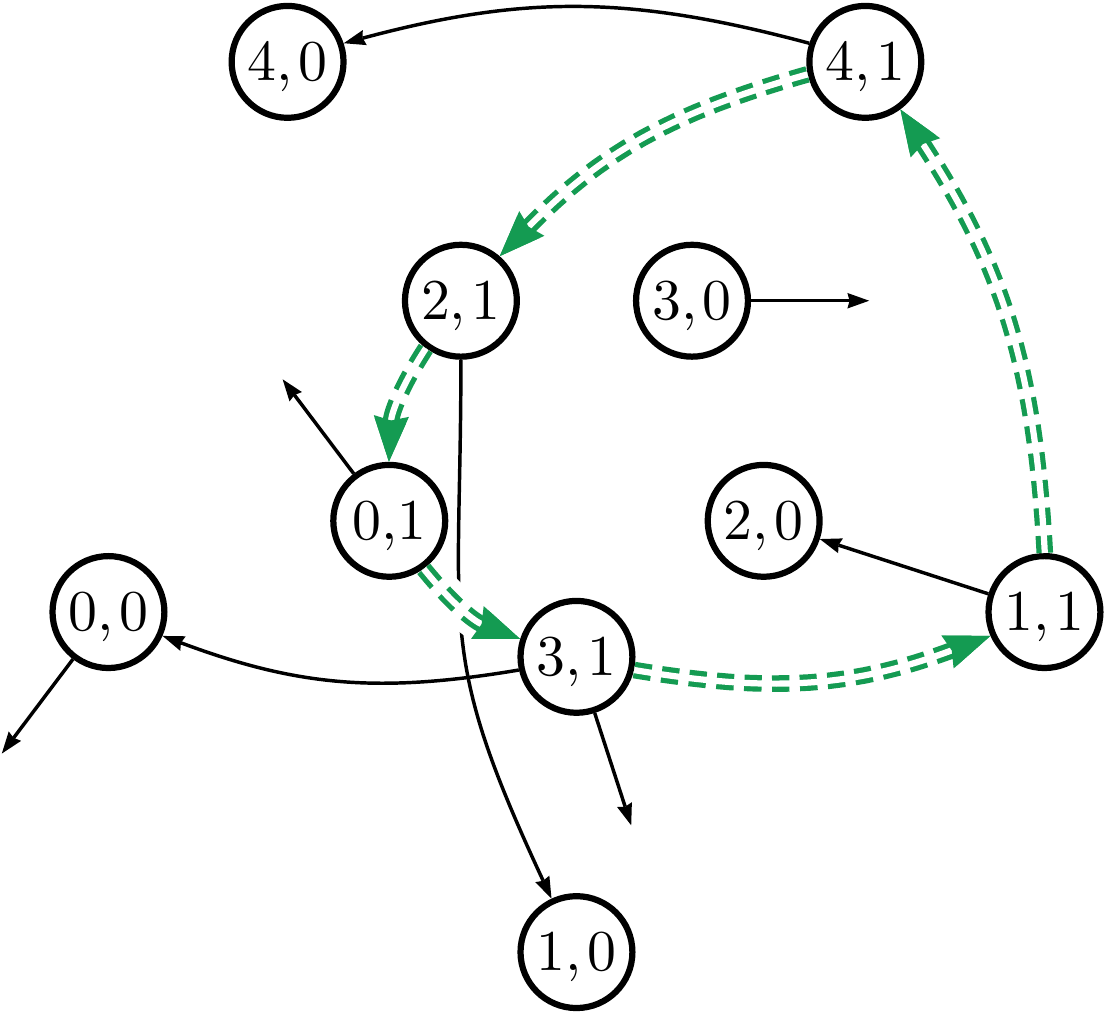}
    \caption{Applying rule (iv) of \rtable{tran-tab}: \\$(s,1)\pathx{g} (s+3,1)$}
  \end{subfigure}

  \vspace{0.8em}
  \begin{subfigure}[t]{0.41\textwidth}
    \centering
    \includegraphics[scale=\AutScale]{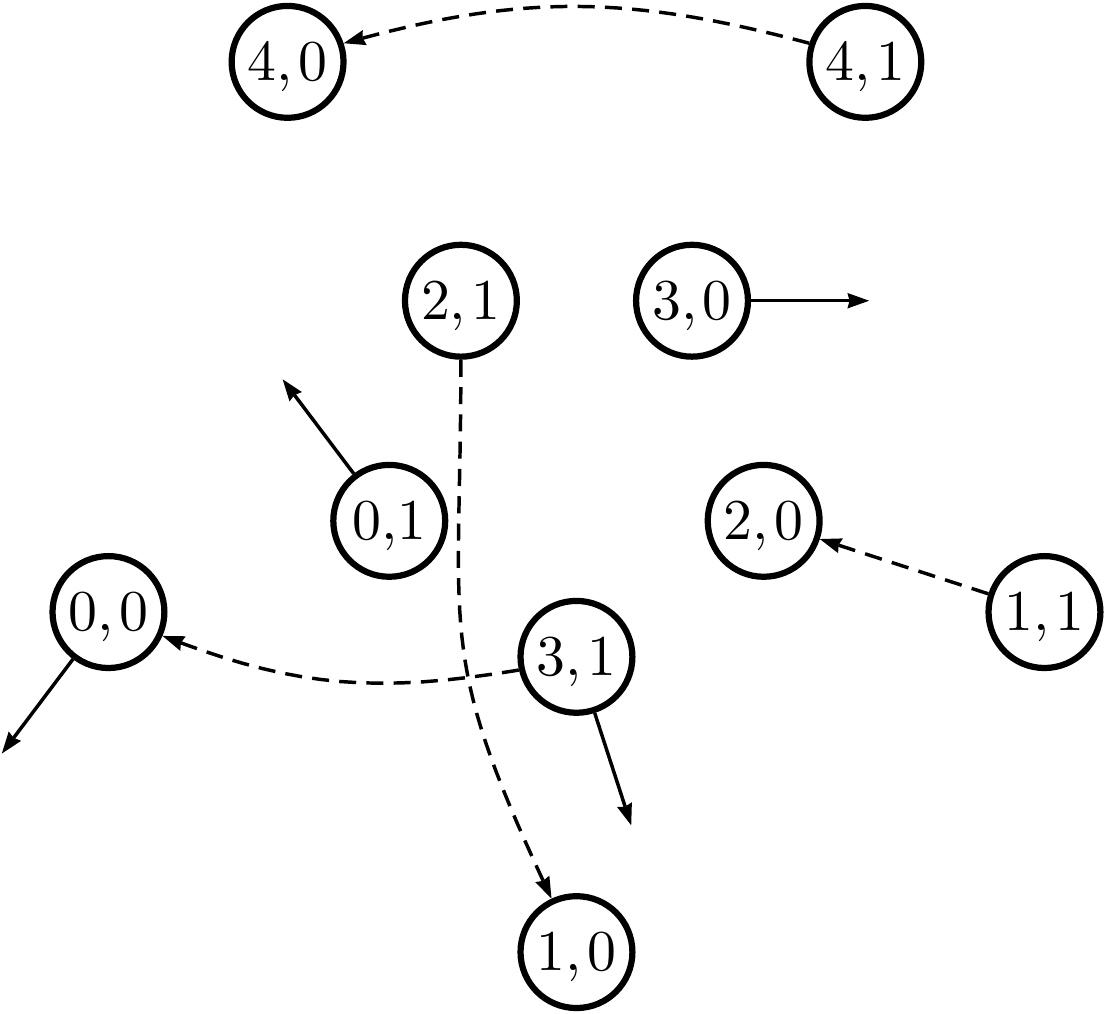}
    \caption{Applying rule (ii) of \rtable{tran-tab}: \\$(s,1)\pathx{0} (4s-2,0)$}
  \end{subfigure}\hspace{15mm}
  \begin{subfigure}[t]{0.41\textwidth}
    \centering
    \includegraphics[scale=\AutScale]{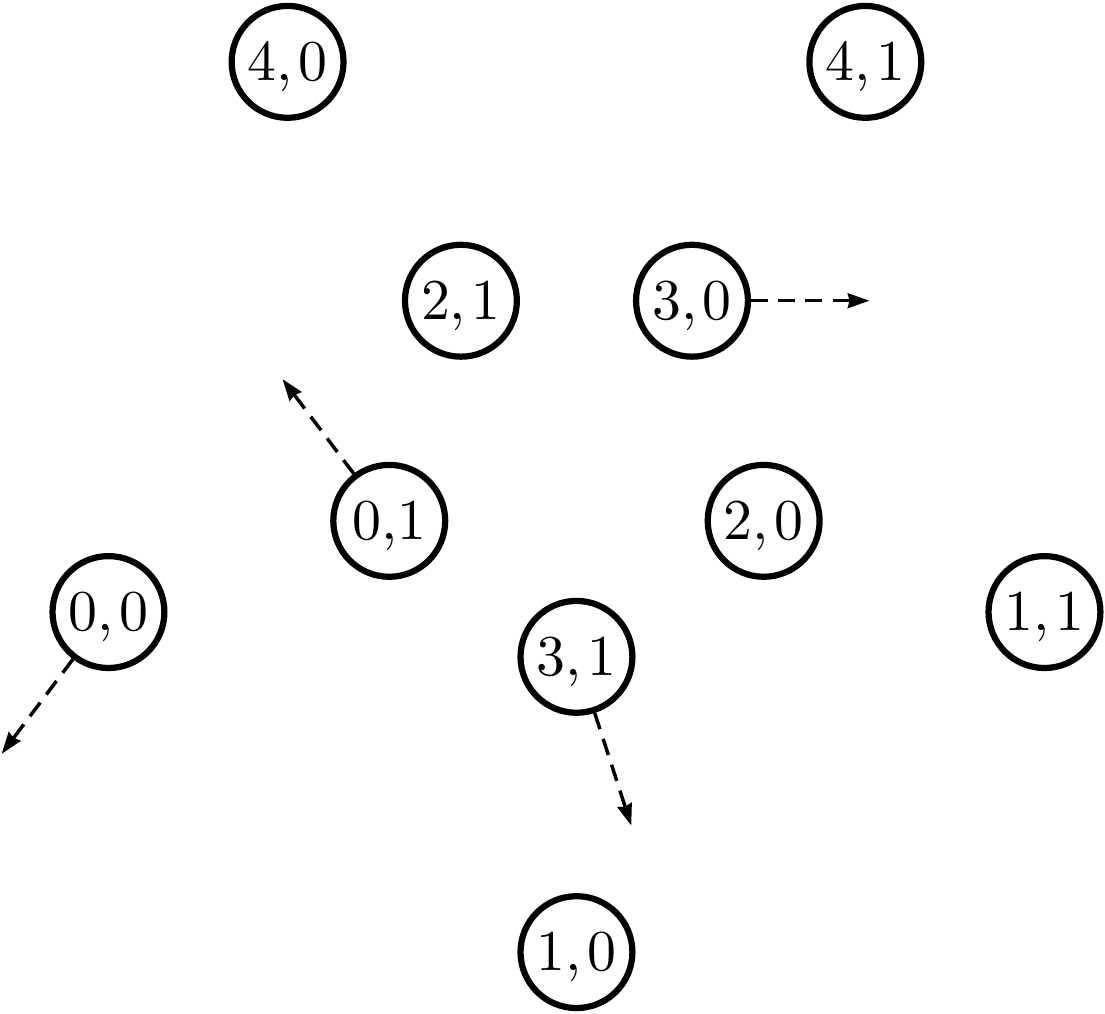}
    \caption{Verifying that~$(s,t)$ is final if and only if $s\in R=\set{0,3}$}
  \end{subfigure}
  \caption{Verifications}\lfigure{verif}
\end{figure}

\paragraph{\trianglebullet~Step 3 (Verifications)}
From \rtheorem{Ahk=A}, if~$\Ac$ is the quotient of a Pascal automaton,
it is isomorphic to~$\Ahk$.
Hence, build~$\Ahk$ using~\rdefinition{Ahk} and tests isomorphism
to~$\Ac'$ with a simple traversal.
%



\begin{runex}
  \rtable{tran-tab} gives the transition function of the
  automaton~$\Ac_{(3,2)}$ (\cf~\rdefinition{Ahk}).
  \rfigure{verif} shows the verification process, the isomorphism is built
  by visiting each transition of~$\Ac'$ and colouring the visited state by
  the corresponding state of~$\Ac_{(3,2)}$.
\end{runex}

\subsection{Linear algorithm to solve Problem~\ref{pb.line-comp}}
\lsection{line-comp-gene}

\begin{thm}\ltheorem{up-line}
  Let~$\Ac$ be a minimal automaton with~$n$ transitions.
  It can be decided in time~$\bigo{b\xmd n}$ whether~$\Ac$ satisfies \upconds.
\end{thm}
\begin{proof}
  A simple traversal is sufficient to check whether~$\Ac$ satisfies~\ref{up.succ-0}.
  Condition~\ref{up.min} is assumed to be satisfied by~$\Ac$.
  The verification of the other conditions require to compute the component graph
  of~$\Ac$; this can be done in time~$\bigo{b\xmd n}$ using
  classical algorithms (\rtheorem{tarj}).
  Verifying~\ref{up.type1} can be done in linear time thanks to the
  algorithm previously presented in \rsection{line-pasc}.
  Verifying~\ref{up.type2} requires a simple test for each of the
  affected \sccs.
  Finally, condition~\ref{up.type2-bis} can be verified in the following
  way.
  Let~$C$ be \ascc{} of type one and~$D$ the \scc of type two
  that descends from it.
  We then define the function~$f$ as follows; it is the only function that may
  realise an  embedding.
  Every state~$x$ in~$C$ is mapped to the unique state~$f(x)$ in~$D$ such
  that
  \begin{equation*}
    x \pathx{1} y \quad\quad\text{and}\quad\quad f(x) \pathx{1} y
  \end{equation*}
  (since~$D$ is a \UPatom automaton, it is a
  group automaton, hence~$y$ and~$f(x)$ are uniquely defined).
  Once~$f$ has been computed, checking whether~$f$ is an embedding function
  can be done in time~$\bigo{b\xmd n}$.
\end{proof}


\section{Conclusion and future work}

The main result of this article is stated again below.
It follows directly from Theorems~\rtheorem*{UP-stru-char}
and~\rtheorem*{up-line}, shown in Sections~\rsection*{UP}
and~\rsection*{line-comp-gene} respectively.

\begin{falsestatement}{\rtheorem{com-plx}}%
  \sttcomplx%
\end{falsestatement}

\begin{falsestatement}{\rcorollary{com-plxnonmin}}
  \sttcomplxnonmin%
\end{falsestatement}

These results almost close the complexity question raised by Honkala's problem, when
one writes representations LSDF\@.
Two improvements are natural: getting rid, in \rtheorem{com-plx},
either of the condition of minimality, or of the condition of determinism.
We are rather optimistic for a positive answer to the first one, by performing
some kind of \emph{partial} minimisation (which would run in linear time).
For instance, the algorithm given in \rsection{line-pasc} solves (a special case)
even if the input automaton is not quite minimal.
On the other hand, devising conditions similar to \upconds* for non-deterministic
automata seems to be much more difficult.
%


%
As for extensions, we are fairly confident that an approach similar to what we
do here can be used for non-standard numeration systems, or at least for a
family of U-systems to be identified.
It would also be interesting to find an equivalent of \upconds for automata
that accept rational subsets of~$\N^d$.
The same questions arise in the case where number representations are written with the
most significant digit first.
We are hopeful that some of them can be addressed by building
upon the recent work of Boigelot et al.~\cite{BoigEtAl17}.
%


\section*{Acknowledgments}
The author is very grateful to the reviewers for pointing out results that
simplified the proofs to a great extent.
The author also warmly thanks Jacques Sakarovitch for suggesting the subject of
this work, and for all the help he provided.
%
%
%

%


\nocite{BertRigo10-b}

\bibliographystyle{alpha}

\bibliography{bibliography,CANT,Alexandrie-abbrevs,Alexandrie-AC,Alexandrie-DF,Alexandrie-GL,Alexandrie-MR,Alexandrie-SZ}

\end{document}